%% file: main.tex
\documentclass[a4paper,reqno]{amsart}
\usepackage{amsfonts}
\usepackage{tikz}
\usepackage{tkz-graph}
\usepackage{accents}
\usepackage{array,calc} 
\usepackage{stmaryrd} 
\usepackage{pgf}
\usetikzlibrary{shapes,cd,arrows,automata,decorations.pathreplacing,angles,quotes,calc}
\usepackage{tkz-euclide}
\usepackage{amsmath,calc,graphicx}
\usepackage{url}
\usepackage[hidelinks]{hyperref}
\usepackage{amssymb}
\usepackage{amsthm}
\usepackage{multirow}
\usepackage{float}
\usepackage{enumerate}
\usepackage{enumitem}
\usepackage{amsrefs}
\usepackage{comment}
\usepackage{todonotes}
\numberwithin{equation}{section}
\numberwithin{figure}{section}
\theoremstyle{plain}

\newtheorem{theorem}{Theorem}[section]
\newtheorem{lemma}[theorem]{Lemma}

\newtheorem{proposition}[theorem]{Proposition}
\newtheorem{definition}[theorem]{Definition}

\theoremstyle{remark}
\newtheorem{remark}[theorem]{Remark}

\newtheorem*{lem*}{\textsc{Lemma}}
\newtheorem*{cor*}{\textsc{Corollary}}
\newtheorem*{exer*}{\textsc{Exercise}}
\newtheorem*{con*}{\textsc{Conjecture}}
\newtheorem*{thm*}{\textsc{Theorem}}

\newcommand{\beq}{\begin{equation}}
\newcommand{\eeq}{\end{equation}}

\makeatletter
\newcommand*\rel@kern[1]{\kern#1\dimexpr\macc@kerna}
\newcommand*\widebar[1]{%
  \begingroup
  \def\mathaccent##1##2{%
    \rel@kern{0.8}%
    \overline{\rel@kern{-0.8}\macc@nucleus\rel@kern{0.2}}%
    \rel@kern{-0.2}%
  }%
  \macc@depth\@ne
  \let\math@bgroup\@empty \let\math@egroup\macc@set@skewchar
  \mathsurround\z@ \frozen@everymath{\mathgroup\macc@group\relax}%
  \macc@set@skewchar\relax
  \let\mathaccentV\macc@nested@a
  \macc@nested@a\relax111{#1}%
  \endgroup
}
\makeatother

\newcommand\Psix{\textrm{P}_{\textrm{VI}}}

\newcommand{\orcidauthorA}{0000-0001-7504-4444}
\newcommand{\orcidauthorB}{0000-0002-0461-7580}

\title{On the crystal limit of the q-difference sixth Painlev\'e equation}
\date{}
\thanks{This research was supported by Australian Research
Council Discovery Project \#DP210100129.}
\author[Nalini Joshi]{Nalini Joshi}
\thanks{NJ's ORCID ID is \orcidauthorA.}
\address{School of Mathematics and Statistics F07, The University of Sydney, NSW 2006, Australia}
\author{Pieter Roffelsen}
\thanks{PR's ORCID ID is \orcidauthorB.}
\email{nalini.joshi@sydney.edu.au}
\email{pieter.roffelsen@sydney.edu.au}
\subjclass[2020]{39A13, 33E17,34M50,39A45,47B39,14J26}

\begin{document}
\begin{abstract}
We consider the Riemann-Hilbert correspondence associated with the $q$-difference sixth Painlev\'e equation in the crystal limit, i.e. $q\to0$, and show two main results. First, the limit of this generically highly transcendental mapping is shown to exist. Second, we show that the limiting map is bi-rational and describe it explicitly.
\end{abstract}
\maketitle


\input{intro}

\input{initial_value_space_and_linear_problem}

\input{monodromy}

\input{conclusion}

\end{document}

%% file: intro.tex
\section{Introduction}
Crystal limits arise in the theory of quantum groups \cite{kashiwara}, which are associated with solvable lattice models in quantum statistical mechanics. As the latter setting gives rise to integrable systems, it is natural to ask whether the combinatorial aspects of the theory of quantum groups, which correspond to the crystal limit $q\to0$, also extend to  integrable $q$-difference equations.

In this paper, we consider this question through the Riemann-Hilbert correspondence associated with a $q$-difference Painlev\'e equation. 
For differential Painlev\'e equations, the study of the Riemann-Hilbert correspondence as $t$ approaches a critical point of the equation has been a major focus of attention in the field \cites{fokasitsbook,itsnovokbook,jimbo82}. However, for $q$-difference  Painlev\'e equations, there are additional limits of interest, in particular in the parameter $q$. We focus on $q\to0$, the crystal limit, here and surprisingly find not only that the limit of the correspondence exists, but also that it is explicitly realised by a bi-rational mapping. 

To be specific, we focus on the $q$-difference sixth Painlev\'e equation (or $q\Psix$). Given $q\in\mathbb C$, $0<|q|<1$, and $\kappa=(\kappa_0,\kappa_t,\kappa_1,\kappa_\infty)\in (\mathbb C^*)^4$, this equation is given by
\begin{align}\label{eq:qpvi}
q\Psix:\ \begin{cases}\  \underline{f}f=\dfrac{(g-\kappa_0\,t)(g-\kappa_0^{-1}t)}{(g-q\,\kappa_\infty)(g-\kappa_\infty^{-1})}, &\\
  \   g\overline{g}=q\dfrac{(f-\kappa_t\,t)(f-\kappa_t^{-1}t)}{(f-\kappa_1)(f-\kappa_1^{-1})}, &
     \end{cases} 
\end{align}
where $f=f(t)$, $g=g(t)$, $\underline{f}=f(t/q)$, $\overline{g}=g(qt)$. The equation is due to Jimbo and Sakai \cite{jimbosakai}\footnote{If the latter's solutions are denoted $(f_{\text{\tiny JS}}, g_{\text{JS}})$, then we have taken $f=f_{\text{\tiny JS}}$ and $g=q\,g_{\text{\tiny JS}}$ and further normalised the equation as described in \cite{jr_qp6}.}
who showed that the mathematical properties of $q\Psix$ closely resemble those of the celebrated sixth Painlev\'e differential equation $\Psix$. It is well known that $q\Psix$ leads to $\Psix$ in the continuum limit $q\to1$.  

For both $q\Psix$ and $\Psix$, with generic parameters, the Riemann-Hilbert correspondence can be realised as a bi-holomorphic mapping between an initial value space and an associated monodromy manifold \cites{inaba2006,rofqpviasymptotics}. For $\Psix$, as well as the other differential Painlev\'e equations, this manifold is an affine cubic surface \cites{jimbo82,putsaito}, while for $q\Psix$ it is an affine Segre surface \cite{jr_qp6}. 
Recently, we showed that all differential Painlev\'e equations have monodromy manifolds given by affine Segre surfaces
\cite{jmr_segre}. 

These earlier results were obtained by taking the continuum limit of $q\Psix$, combined with confluence limits of the Painlev\'e equations. In this paper, we consider a very different limit ($q\to0$). Our main result is Theorem \ref{thm:main}, which shows that the Riemann-Hilbert correspondence becomes a bi-rational mapping under this limit.

\subsection{Outline}
The paper is organised as follows. In \S\ref{s:ivs_lin} we study the crystal limit of the initial value space of $q\Psix$ and its associated linear $q$-difference problem. In \S\ref{s:segre_RHP}, we consider the crystal limit of the Riemann-Hilbert correspondence and associated Segre surface. Finally, we end the paper with a conclusion in \S\ref{s:conc}.


\subsection{Notation}
Throughout the paper, the crystal limit of any given object $X$, say, that depends on $q$, is denoted by a superscript diamond. That is, \begin{equation*}
    X\xrightarrow{q\rightarrow 0} X^\diamond.
\end{equation*}
All manifolds and varieties in this paper are over $\mathbb{C}$, and we will write $\mathbb{P}^n$ for $n$-dimensional complex projective space. Furthermore, we denote by $\sigma_3$ the Pauli spin matrix
\begin{equation*}
    \sigma_3=\begin{bmatrix}
        1 & 0\\
        0 & -1
    \end{bmatrix}.
\end{equation*}

%% file: initial_value_space_and_linear_problem.tex
\section{Crystal limit of the initial value space and linear problem}\label{s:ivs_lin}

In this section, we study the the initial value space and the linear problem of $q\Psix$ in the limit $q\to0$. The limit of the initial value space is considered in Section \ref{subsec:initialvalue}. This is followed by Section \ref{sec:linear_system_geom}, where the geometry of the linear system associated with $q\Psix$ is studied. In Section \ref{sec:canonical_sol}, we determine the crystal limits of canonical solutions of the linear system. These canonical solutions give rise to a connection matrix that lies in a monodromy manifold. The limit of the connection matrix is described in Section \ref{sec:connection_matrix}.

We start by describing conditions on the parameters relevant in the Riemann-Hilbert approach to $q\Psix$, see \cite{jr_qp6}, and what they are replaced by in the crystal limit. First, we have the {\em non-resonance} conditions given by
\begin{equation}\label{eq:non_res}
\kappa_0^2,\kappa_t^2,\kappa_1^2,\kappa_\infty^2\notin q^\mathbb{Z},\qquad (\kappa_t\kappa_1)^{\pm 1},
(\kappa_t/\kappa_1)^{\pm 1}\notin t\, q^\mathbb{Z}.
\end{equation}
Second, the {\em non-splitting} conditions are
\begin{subequations}\label{eq:intro_irreducibleparameter}
\begin{align}
&\kappa_0^{\epsilon_0}\kappa_t^{\epsilon_t} \kappa_1^{\epsilon_1} \kappa_\infty^{\epsilon_\infty}\notin q^\mathbb{Z},\label{eq:intro_irreducibleparameter1}\\
&\kappa_0^{\epsilon_0} \kappa_\infty^{\epsilon_\infty}\notin t\,q^\mathbb{Z},\label{eq:intro_irreducibleparameter2}
\end{align}
\end{subequations}
where $\epsilon_j\in\{\pm 1\}$, $j=0,t,1,\infty$.

In the limit $q\rightarrow 0$, we replace these conditions by
\begin{equation}\label{eq:non_rescrystal}
\kappa_0^2,\kappa_t^2,\kappa_1^2,\kappa_\infty^2\neq 1,\qquad (\kappa_t\kappa_1)^{\pm 1},
(\kappa_t/\kappa_1)^{\pm 1}\neq t,
\end{equation}
and
\begin{subequations}\label{eq:intro_irreducibleparametercrystal}
\begin{align}
&\kappa_0^{\epsilon_0}\kappa_t^{\epsilon_t} \kappa_1^{\epsilon_1} \kappa_\infty^{\epsilon_\infty}\neq 1,\label{eq:intro_irreducibleparameter1crystal}\\
&\kappa_0^{\epsilon_0} \kappa_\infty^{\epsilon_\infty}\neq t,\label{eq:intro_irreducibleparameter2crystal}
\end{align}
\end{subequations}
where $\epsilon_j\in\{\pm 1\}$, $j=0,t,1,\infty$. We will assume conditions \eqref{eq:non_rescrystal} and \eqref{eq:intro_irreducibleparametercrystal} throughout the paper.

\subsection{Crystal limit of the initial value space}\label{subsec:initialvalue}
The initial value space of $q\Psix$ is obtained by blowing up the compact space $\{(f, g)\in\mathbb P^1\times \mathbb P^1\}$ at the 8 base points \cite{s:01}
\begin{equation}
    \begin{aligned}
    &b_1=(0,\kappa_0^{+1}t), & &b_3=(\kappa_t^{+1}t,0), & &b_5=(\kappa_1^{+1},\infty), & &b_7=(\infty,\kappa_\infty^{-1}),\\
    &b_2=(0,\kappa_0^{-1}t), & &b_4=(\kappa_t^{-1}t,0), & &b_6=(\kappa_1^{-1},\infty), & &b_8=(\infty,q\,\kappa_\infty^{+ 1}).
\end{aligned}\label{eq:basepoints}
\end{equation}
Denoting the resulting space by $\overline{\mathcal{X}}_t$, the dynamical system \eqref{eq:qpvi} lifts to an isomorphism from $\overline{\mathcal{X}}_t$ to $\overline{\mathcal{X}}_{qt}$. The monodromy mapping constructed in \cite{jr_qp6} assumes that solutions of $q\Psix$ take at least one value in $\mathbb{C}^*\times \mathbb{C}^*$. We are therefore led to the following definition.
\begin{definition}\label{def:initial_space}
  The initial value space of $q\Psix$ is defined as the open surface
\begin{equation*}
    \mathcal{X}_t:=\overline{\mathcal{X}}_t\setminus D_t,
\end{equation*}  
where $D_t$ is the union of the strict transforms of the curves $f=0$, $f=\infty$, $g=0$ and $g=\infty$.
\end{definition}

Denote by $E_k$ the exceptional line corresponding to $b_k$ in $\overline{\mathcal{X}}_t$, $1\leq k\leq 8$. 
The parts of the exceptional lines away from $D_t$ are explicitly parametrised by
\begin{equation}\label{eq:exceptional_para}
    E_k=\{v_k\in\mathbb{C}:u_k=0\}\quad (1\leq k\leq 8),
\end{equation}
where each of the pairs of coordinates $(u_k,v_k)$, $1\leq k\leq 8$, comes from a bi-rational change of variables,

\begin{align*}
   &\begin{cases}
        f=u_1, &\\
        g=\kappa_0\,t+u_1\,v_1, &
    \end{cases} &
    &\begin{cases}
        f=u_2, &\\
        g=\kappa_0^{-1}\,t+u_2\,v_2, &
    \end{cases}\\
  &\begin{cases}
        f=\kappa_t\,t+u_3\,v_3, &\\
        g=u_3, &
    \end{cases} &
    &\begin{cases}
        f=\kappa_t^{-1}\,t+u_4\,v_4, &\\
        g=u_4, &
    \end{cases}\\
      &\begin{cases}
        f=\kappa_1+u_5\,v_5, &\\
        g=u_5^{-1}, &
    \end{cases} &
    &\begin{cases}
        f=\kappa_1^{-1}+u_6\,v_6, &\\
        g=u_6^{-1}, &
    \end{cases}\\
      &\begin{cases}
        f=u_7^{-1}, &\\
        g=\kappa_\infty^{-1}+u_7\, v_7, &
    \end{cases} &
    &\begin{cases}
        f=u_8^{-1},&\\
        g=q\,\kappa_\infty+u_8\,v_8. &
    \end{cases}
\end{align*}

Whilst the surface $\mathcal{X}_t$ remains well-defined at $q=0$, it is geometrically distinct from the case $q\neq 0$,  since
the base-point configuration changes in this limit,
\begin{equation*}
    b_8=(\infty,q \, \kappa_\infty)\xrightarrow{q\rightarrow 0} (\infty,0),
\end{equation*}
see Figure \ref{fig:degeneration_initial}. 
This motivates our definition of a subspace of the initial value space, which arises by blowing up all the base points except for $b_8$.

\begin{definition}\label{defi:mathfrakX}
    Let $\overline{\mathfrak{X}}_t$ be the compact rational surface obtained by blowing up $\mathbb{P}^1\times \mathbb{P}^1$ at the points $b_k$, $1\leq k\leq 7$. We define $\mathfrak{X}_t$ as the open surface obtained by
removing from $\overline{\mathfrak{X}}_t$ the strict transform under these blow-ups of the union of the curves $f=0$, $f=\infty$, $g=0$ and $g=\infty$.
\end{definition}
Note that 
\begin{equation}\label{eq:subspaceinitialvalue}
    \mathfrak{X}_t=\mathcal{X}_t\setminus E_8,
\end{equation}
and that, contrary to $\mathcal{X}_t$,  $\mathfrak{X}_t$ does not depend on $q$.

\begin{figure}[H]
\centering
	\begin{tikzpicture}[scale=.85,>=stealth,basept/.style={circle, draw=red!100, fill=red!100, thick, inner sep=0pt,minimum size=1.2mm}]
		\begin{scope}[xshift = -3.5cm]
			\draw [black, line width = 1pt] 	(4.1,2.5) 	-- (-0.5,2.5)	node [left]  {$g=\infty$} node[pos=0, right] {};
			\draw [black, line width = 1pt] 	(0,3) -- (0,-1)			node [below] {$f=0$}  node[pos=0, above, xshift=-7pt] {} ;
			\draw [black, line width = 1pt] 	(3.6,3) -- (3.6,-1)		node [below]  {$f=\infty$} node[pos=0, above, xshift=7pt] {};
			\draw [black, line width = 1pt] 	(4.1,-.5) 	-- (-0.5,-0.5)	node [left]  {$g=0$} node[pos=0, right] {};

			\node (p1) at (0,0.5) [basept,label={[xshift=-10pt, yshift = -10 pt] $b_{1}$}] {};
			\node (p2) at (0,1.5) [basept,label={[xshift=-10pt, yshift = -10 pt] $b_{2}$}] {};
			\node (p3) at (1.2,-0.5) [basept,label={[yshift=-20pt, xshift=0pt] $b_{3}$}] {};
			\node (p4) at (2.4,-0.5) [basept,label={[xshift=0pt, yshift = -20 pt] $b_{4}$}] {};
			\node (p5) at (1.2,2.5) [basept,label={[xshift=0pt,yshift=0pt] $b_{5}$}] {};
			\node (p6) at (2.4,2.5) [basept,label={[xshift=0pt, yshift = 0 pt] $b_{6}$}] {};
			\node (p7) at (3.6,1.5) [basept,label={[xshift=10pt, yshift = -10 pt] $b_{7}$}] {};
			\node (p8) at (3.6,0.5) [basept,label={[xshift=10pt, yshift = -10 pt] $b_{8}$}] {};

			\node (P1P1) at (1.8, -1.8) [below] {$\mathbb{P}^1 \times \mathbb{P}^1$};
		\end{scope}
	
		\draw [<-] (3.5,1)--(1.5,1) node[pos=0.5, below] {$q\rightarrow 0$};
	
		\begin{scope}[xshift = 6cm, yshift= 0cm]
				\draw [black, line width = 1pt] 	(4.1,2.5) 	-- (-0.5,2.5)	node [left]  {$g=\infty$} node[pos=0, right] {};
			\draw [black, line width = 1pt] 	(0,3) -- (0,-1)			node [below] {$f=0$}  node[pos=0, above, xshift=-7pt] {} ;
			\draw [black, line width = 1pt] 	(3.6,3) -- (3.6,-1)		node [below]  {$f=\infty$} node[pos=0, above, xshift=7pt] {};
			\draw [black, line width = 1pt] 	(4.1,-.5) 	-- (-0.5,-0.5)	node [left]  {$g=0$} node[pos=0, right] {};

			\node (p1) at (0,0.5) [basept,label={[xshift=-10pt, yshift = -10 pt] $b_{1}$}] {};
			\node (p2) at (0,1.5) [basept,label={[xshift=-10pt, yshift = -10 pt] $b_{2}$}] {};
			\node (p3) at (1.2,-0.5) [basept,label={[yshift=-20pt, xshift=0pt] $b_{3}$}] {};
			\node (p4) at (2.4,-0.5) [basept,label={[xshift=0pt, yshift = -20 pt] $b_{4}$}] {};
			\node (p5) at (1.2,2.5) [basept,label={[xshift=0pt,yshift=0pt] $b_{5}$}] {};
			\node (p6) at (2.4,2.5) [basept,label={[xshift=0pt, yshift = 0 pt] $b_{6}$}] {};
			\node (p7) at (3.6,1.5) [basept,label={[xshift=10pt, yshift = -10 pt] $b_{7}$}] {};
			\node (p8) at (3.6,-0.5) [basept,label={[xshift=10pt, yshift = -0 pt] $b_{8}$}] {};

			\node (P1P1) at (1.8, -1.8) [below] {$\mathbb{P}^1 \times \mathbb{P}^1$};
		\end{scope}
	\end{tikzpicture}
	\caption{Degeneration of base-point configuration in \eqref{eq:basepoints} under the crystal limit.}
	\label{fig:degeneration_initial}
\end{figure}
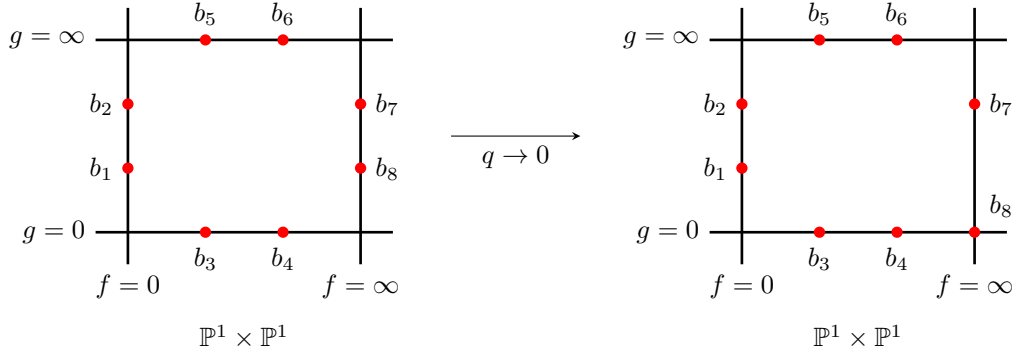

\begin{remark}\label{rem:symmetry}
  The symmetry group of the initial value space $\mathcal{X}_t$ is given by the extended affine Weyl group of type $D_5^{(1)}$ \cite{s:01}. We note that translations on the $D_5^{(1)}$ lattice become ill-defined as $q\rightarrow 0$. The limiting finite group remains to be described.
\end{remark}

\begin{remark}
We remark that taking the crystal limit of $q\Psix$ in the original coordinates $(f_{\text{\tiny JS}}, g_{\text{\tiny JS}})$ would have led 3 of its 8 base points to each approach a corner of the coordinate diagram of $\mathbb P^1\times \mathbb P^1$ (see Figure \ref{fig:degeneration_initial}). Our choice of coordinates for the normalised form of $q\Psix$ instead leads only one of the base points to approach a corner, which we have chosen to be $b_8$. (See Remark \ref{rem:special} for the Riemann-Hilbert point of view of its distinguished properties.) 

At least one base point must approach a corner point, as shown by the following relation among the $8$ base points defined in \eqref{eq:basepoints},
\begin{equation*}
    \frac{(b_1)_2(b_2)_2}{(b_7)_2(b_8)_2}=q \frac{(b_3)_1(b_4)_1}{(b_5)_1(b_6)_1},
\end{equation*}
where $(b_k)_j$ is the value of the $j$th component of $b_k$, for $j=1,2$ and $1\leq k\leq 8$. This relation is invariant under multiplicative scalings of the dependent variables, the independent variable and the parameters of $q\Psix$. The scaling we have chosen is optimal in the sense that the crystal limit leads to only one base point being distinguished in this way.
\end{remark}

\subsection{Geometry of the linear system and crystal limit}\label{sec:linear_system_geom}
Given $t$ and $\kappa\in(\mathbb{C}^*)^4$, the Jimbo-Sakai linear problem \cite{jimbosakai} (rescaled as in \cite[\S 3.1]{jr_qp6}) is  given by
\begin{align}
    Y(qz)&=A(z)Y(z),\label{eq:linear_problem}\\
    A(z)&=A_0+z\, A_1+z^2 A_2,\label{eq:matrix_polynomial}
\end{align}
where $A(z)$ is a $2\times 2$ matrix polynomial in $z$, with determinant given by
\begin{equation}\label{eq:detA}
|A(z)|=(z-\kappa_t^{+1}t)(z-\kappa_t^{-1}t)(z-\kappa_1^{+1})(z-\kappa_1^{-1}),
\end{equation}
and
\begin{equation}
    A_0=H\begin{bmatrix}
    \kappa_0^{+1}t & 0\\
    0 & \kappa_0^{-1}t \end{bmatrix}H^{-1},\quad  A_2=\begin{bmatrix}
    \kappa_\infty^{+1} & 0\\
    0 & \kappa_\infty^{-1} \end{bmatrix}.\label{eq:diagonal}
\end{equation}
for an $H\in GL_2(\mathbb{C})$.

The coefficient matrix is given by
\begin{equation}\label{eq:param1}
    A(z)=\begin{bmatrix}
    \kappa_\infty ((z-f)(z-\alpha)+g_1) & \kappa_\infty^{-1} w(z-f)\\
    \kappa_\infty w^{-1}(\gamma\, z+\delta) & \kappa_\infty^{-1}((z-f)(z-\beta)+g_2)
    \end{bmatrix},
\end{equation}
where
\begin{align*}
g_1&=\kappa_\infty^{-1}(f-\kappa_t t)(f-\kappa_t^{-1}t)g^{-1},\\ 
g_2&=\kappa_\infty(f-\kappa_1)(f-\kappa_1^{-1})g,\nonumber
\end{align*}
and, temporarily using the notation $\mathring{\kappa}_j=\kappa_j+\kappa_j^{-1}$, $j=0,t,1,\infty$,
\begin{align*}
\alpha&=\frac{1}{(1-\kappa_\infty^2)f}
\left(\kappa_\infty^2 g_1-\kappa_\infty\mathring{\kappa}_0t+g_2+(\mathring{\kappa}_tt+\mathring{\kappa}_1)f-2 f^2\right),\\
\beta&=\frac{1}{(\kappa_\infty^2-1)f}
\left(\kappa_\infty^2 g_1-\kappa_\infty\mathring{\kappa}_0t+g_2+\kappa_\infty^2(\mathring{\kappa}_tt+\mathring{\kappa}_1)f-2 \kappa_\infty^2f^2\right),\\
\gamma&=g_1+g_2+f^2+2(\alpha+\beta)f+\alpha\beta -(t^2+\mathring{\kappa}_t\mathring{\kappa}_1t+1),\\           
\delta&=f^{-1}(t^2-(g_1+\alpha f)(g_2+\beta f)).
\end{align*}

Note that the coefficient matrix $A(z)$ is independent of $q$,  In particular, the matrix $H$, diagonalising $A_0$ in \eqref{eq:diagonal}, can and will be chosen independent of $q$.

The linear system has a large number of symmetries induced by gauge transformations \cite{ormerod},
\begin{equation}\label{eq:gauge}
Y(z)\mapsto \tilde{Y}(z)=G(z)Y(z),\qquad A(z)\mapsto \tilde{A}(z)=G(q\,z)A(z)G(z)^{-1},
\end{equation}
where $G(z)\in GL_2(\mathbb{C}(z))$ is such that $\tilde{A}(z)$ is again a matrix polynomial of the form \eqref{eq:matrix_polynomial}.
One of these symmetries corresponds to the time-evolution of $q\Psix$ \cite{jimbosakai},
\begin{equation*}
    G(z)=B(z):=\displaystyle\frac{z^2 I+z\,B_0}{(z-q\,\kappa_t^{+1}t)(z-q\,\kappa_t^{-1}t)},
\end{equation*}
where
\begin{equation*}
    B_0= \begin{bmatrix}\displaystyle
    \frac{q}{1-q}(\widebar{f}+\widebar{\beta}-f-\beta) &\displaystyle -\frac{q\,(\widebar{w}-w)}{q\,\kappa_\infty^2-1} \\
    \displaystyle\frac{q\,\kappa_\infty^2}{\kappa_\infty^2-q}\left(\frac{\widebar{\gamma}}{\widebar{w}}-\frac{\gamma}{w}\right) & \displaystyle\frac{q}{1-q}(\widebar{f}+\widebar{\alpha}-f-\alpha)
    \end{bmatrix}.
\end{equation*}
In other words,
\begin{equation}\label{eq:time_evolution}
    \widebar{A}(z)=B(q\,z)A(z)B(z)^{-1},
\end{equation}
where $\widebar{A}(z)$ denotes the coefficient matrix with $\widebar{t}=q\,t$, coordinates $\widebar{f}$, $\widebar{g}$ defined by the $q\Psix$ evolution \eqref{eq:qpvi} and
\begin{equation*}
    \widebar{w}=w\, \frac{1-q\,\kappa_\infty\,\widebar{g}}{1-\kappa_\infty^{-1}\,\widebar{g}}.
\end{equation*}

For fixed $(f,g,w)$, as $q\rightarrow 0$, we have
\begin{equation*}
    \widebar{f}=\mathcal{O}(q),\quad
    \widebar{g}=\mathcal{O}(q),\quad
    \widebar{w}=\mathcal{O}(1),
\end{equation*}
and correspondingly
\begin{equation*}
    \widebar{\alpha}=\mathcal{O}(1),\quad
    \widebar{\beta}=\mathcal{O}(1),\quad
    \widebar{\gamma}=\mathcal{O}(1),\quad
    \widebar{\delta}=\mathcal{O}(q),
\end{equation*}
so that
\begin{equation*}
   B_0=q\, B_1+\mathcal{O}(q^2),
\end{equation*}
for some matrix $B_1$ with rational entries in $(f,g,w)$. Therefore
\begin{equation*}
    B(z)=I+\mathcal{O}(q),\qquad B(q\,z)=\displaystyle\frac{z^2 I+z\,B_1}{(z-\kappa_t^{+1}t)(z-\kappa_t^{-1}t)}+\mathcal{O}(q).
\end{equation*}
This means that
\begin{equation*}
    \widebar{A}(z)\xrightarrow[]{q\rightarrow 0}\widebar{A}^\diamond(z),\qquad
    \widebar{A}^\diamond(z):=\displaystyle\frac{z^2 I+z\,B_1}{(z-\kappa_t^{+1}t)(z-\kappa_t^{-1}t)} A(z).
\end{equation*}
Now, the limiting coefficient matrix $\widebar{A}^\diamond(z)$ takes the form
\begin{equation*}
    \widebar{A}^\diamond(z)=z\, \widebar{A}_1^\diamond+z^2\,\widebar{A}_2^\diamond,\qquad \widebar{A}_2^\diamond=\begin{bmatrix}
    \kappa_\infty^{+1} & 0\\
    0 & \kappa_\infty^{-1} \end{bmatrix},
\end{equation*}
with determinant $|\widebar{A}^\diamond(z)|=z^2(z-\kappa_1)(z-\kappa_1^{-1})$. Upon removing the overall factor $z$, the evolved linear problem has a coefficient matrix that lies in the hypergeometric class \cite{lecaine}. This is a manifestation of the metamorphosis of the dynamical system in the crystal limit.

We now consider properties of the space of coefficient matrices,
\begin{equation*}
    \mathcal{A}_t=\{A(z)\in GL_2(\mathbb{C}[z]) \text{ satisfying \eqref{eq:matrix_polynomial}, \eqref{eq:detA},\eqref{eq:diagonal}}\}.
\end{equation*}
We denote entries of the coefficients in \eqref{eq:diagonal} by
\begin{equation}\label{eq:A0A1notation}
    A_j=\begin{bmatrix}
        a_j & b_j\\
        c_j & d_j
    \end{bmatrix},
\end{equation}
$j=0,1$. Since $A_2$ is fixed, we consider the entries of $A_0$ and $A_1$ as variables. This implies that the $\mathcal{A}_t$ is an affine algebraic set, as the conditions in equations \eqref{eq:detA} and \eqref{eq:diagonal} are all polynomial conditions.


The auxiliary variable $w$ simply parametrises the freedom of rescaling $A(z)$ by conjugation with a diagonal matrix. Let $\mathcal{A}_t/\hspace{-0.5mm}\sim$ denote the algebraic quotient  of $\mathcal{A}_t$ with respect to conjugation by diagonals. Introducing variables $u_{ij}=b_i\,c_j$, $0\leq i,j\leq 1$, $\mathcal{A}_t/\hspace{-0.5mm}\sim$  is explicitly given by the vanishing locus of the following six polynomials in the eight variables $a_i,d_j,u_{ij}$, $0\leq i,j\leq 1$;
\begin{equation*}
u_{00}u_{11}-u_{01}u_{10},\quad a_0+d_0-t(\kappa_0+\kappa_0^{-1}),
\end{equation*}
and
\begin{equation*}
\Delta(\kappa_t t),\quad \Delta(\kappa_t^{-1} t),\quad \Delta(\kappa_1),\quad \Delta(\kappa_1^{-1}),
\end{equation*}
where
\begin{equation*}
    \Delta(z):=(a_0+a_1 z+\kappa_\infty z^2)(d_0+d_1 z+\kappa_\infty^{-1} z^2)-(u_{00}+(u_{01}+u_{10})z+u_{11}z^2).
\end{equation*}
Here, the vanishing of the first follows from the definition of the $u_{ij}$, the vanishing of the second is a consequence of \eqref{eq:diagonal} and the vanishing of the remaining four follow from \eqref{eq:detA}. Equation \eqref{eq:diagonal} further implies $a_0d_0-u_{00}-t^2=0$ and in this regard we note that the left-hand side is contained in the ideal generated by the six polynomials above.

We are now in a position to describe the mapping between the initial value space and the quotient space of coefficient matrices. 

\begin{lemma}\label{lem:iso_init_to_linear}
The rational mapping 
\begin{equation}\label{eq:parametrisationmap}
 \mathfrak{X}_t    \rightarrow \mathcal{A}_t/\hspace{-0.5mm}\sim,
\end{equation}
realised through the parametrisation $(f,g)\mapsto [A(z)] $,
is an isomorphism.
\end{lemma}
\begin{proof}
We will start by showing that \eqref{eq:parametrisationmap} is a regular map.
Note that formula \eqref{eq:param1} shows that the mapping is regular away from the exceptional lines $E_k$, $1\leq k\leq 7$.

For $1\leq k\leq 7$, we use local coordinates to study the mapping in the neighbourhood of the exceptional lines $E_k$; see equation \eqref{eq:exceptional_para}. Considering $k=1$, we have
\begin{align*}
a_0=\,&\kappa_0^{-1}t+\mathcal{O}(u_1),\\
d_0=\,&\kappa_0^{+1}t+\mathcal{O}(u_1),\\
a_1=\,&\frac{1}{\kappa_\infty^{-2}-1}\left(t\,\kappa_0\,\mathring{\kappa}_1+\kappa_0^{-1} \mathring{\kappa}_t-\kappa_\infty^{-1}(t\, \mathring{\kappa}_t+\mathring{\kappa}_1)+(\kappa_0^{-2}-1)v_1\right)+\mathcal{O}(u_1),\\
d_1=\,&\frac{1}{\kappa_\infty^{+2}-1}\left(t\,\kappa_0\,\mathring{\kappa}_1+\kappa_0^{-1} \mathring{\kappa}_t-\kappa_\infty^{+1}(t\, \mathring{\kappa}_t+\mathring{\kappa}_1)+(\kappa_0^{-2}-1)v_1\right)+\mathcal{O}(u_1),\\
b_0\, c_0=\,&\mathcal{O}(u_1),\\
b_0\, c_1=\,&\mathcal{O}(u_1),\\
b_1\, c_0=\,&\frac{(\kappa_0^{-2}-1)t}{\kappa_\infty^2-1}\left((t\,\kappa_0\,\kappa_\infty-1)(\kappa_0\,\kappa_\infty\,\mathring{\kappa}_1-\mathring{\kappa}_t)-(\kappa_0\,\kappa_\infty^2-\kappa_0^{-1})v_1\right)+\mathcal{O}(u_1),\\
b_1\, c_1=\,&\frac{t\,\kappa_0\,\kappa_\infty-1}{\kappa_0^2(\kappa_\infty^2-1)^2}(\kappa_\infty-t\,\kappa_0)(\kappa_0\,\kappa_\infty\,\mathring{\kappa}_1-\mathring{\kappa}_t )(\kappa_0\,\mathring{\kappa}_1-\kappa_\infty\, \mathring{\kappa}_t)\\
&+(t\,\kappa_0\,\kappa_\infty-1)(1-\kappa_0^{-1}\kappa_\infty^{-1}t)-\frac{\kappa_\infty^2(\kappa_0^2-1)^2}{\kappa_0^4(\kappa_\infty^2-1)^2}v_1^2\\
&-\frac{\kappa_\infty^3(\kappa_0^2-1)}{\kappa_0^3(\kappa_\infty^2-1)^2}\Big{[}(2\,t\,\kappa_0-\mathring{\kappa}_\infty)\mathring{\kappa}_t+
(t\,\kappa_0\,\mathring{\kappa}_\infty-2)\mathring{\kappa}_1
\Big{]}v_1+\mathcal{O}(u_1),
\end{align*}
as $u_1\rightarrow 0$. This means that the mapping \eqref{eq:parametrisationmap} is regular around the exceptional curve $E_1$. Similarly, we see that the mapping is regular around the other exceptional curves $E_k$, $2\leq k\leq 7$. It follows that the mapping  \eqref{eq:parametrisationmap} is regular.

Points in the exceptional curves $E_1$ and $E_2$ are mapped to classes of coefficient matrices $[A(z)]$ in $\mathcal{A}_t/\sim$, with respectively
\begin{equation*}
    A(0)=\begin{bmatrix} \kappa_0^{-1} t & 0\\
    * & \kappa_0^{+1} t
    \end{bmatrix},\quad    A(0)=\begin{bmatrix} \kappa_0^{+1} t & 0\\
    * & \kappa_0^{-1} t
    \end{bmatrix},
\end{equation*}
 points in $E_3$ and $E_4$ are  mapped to classes
 of coefficient matrices $[A(z)]$ with respectively
\begin{equation*}
    A(\kappa_t\, t)=\begin{bmatrix} * & 0\\
    * & 0
    \end{bmatrix},\quad     A(\kappa_t^{-1} t)=\begin{bmatrix} * & 0\\
    * & 0
    \end{bmatrix},
\end{equation*}
 points in $E_5$ and $E_6$ are  mapped to classes
 of coefficient matrices $[A(z)]$ with respectively
\begin{equation*}
    A(\kappa_1)=\begin{bmatrix} 0 & 0\\
    * & *
    \end{bmatrix},\quad     A(\kappa_1^{-1})=\begin{bmatrix} 0 & 0\\
    * & *
    \end{bmatrix},
\end{equation*}
and points in $E_7$ are mapped to classes of coefficient matrices $[A(z)]$  with $A_{12}(z)\equiv c$ constant in $z$. We note that, to verify the last assertion, one needs to make use of the freedom of conjugation by diagonals, as some of coefficients of $A(z)$ diverge in the limit $u_7\rightarrow 0$ without scaling $w$.
Namely, one sets $w=u_7\, \widetilde{w}$ in \eqref{eq:param1}, so that the limit $u_7\rightarrow 0$ of $A(z)$ is well-defined.

As follows from the parametrisation \eqref{eq:param1}, the mapping  \eqref{eq:parametrisationmap} has a rational inverse, given by
\begin{subequations}\label{eq:rational_inverse}
\begin{align}
     f&=-\frac{u_{00}}{u_{10}}=-\frac{u_{01}}{u_{11}},\label{eq:fform}\\
     g&=\frac{(f-\kappa_t t)(f-\kappa_t^{-1}t)}{a_0+a_1 f+\kappa_\infty f^2}=\frac{d_0+d_1f+\kappa_\infty^{-1}f^2}{(f-\kappa_1)(f-\kappa_1^{-1})}.\label{eq:gform}
\end{align}
\end{subequations}
Due to the non-splitting conditions \eqref{eq:intro_irreducibleparameter}, it is impossible for $u_{00}$, $u_{01}$, $u_{10}$ and $u_{11}$ to be zero simultaneously on $\mathcal{A}_t/\hspace{-0.5mm}\sim$, hence
\begin{equation*}
    \mathcal{A}_t/\hspace{-0.5mm}\sim\rightarrow \mathbb{P}^1,[A(z)]\mapsto f
\end{equation*}
is a regular map. Similarly, by the non-resonance conditions \eqref{eq:non_res}, it is impossible for both the numerators and both the denominators in \eqref{eq:gform} to be simultaneously zero on $\mathcal{A}_t/\hspace{-0.5mm}\sim$. Therefore,
\begin{equation*}
    \mathcal{A}_t/\hspace{-0.5mm}\sim\rightarrow \mathbb{P}^1,[A(z)]\mapsto g
\end{equation*}
is also a regular map. 

By distinguishing between the seven cases above and the generic case, one checks that this induces a unique regular map into $\mathfrak{X}_t$,
\begin{equation*}
       \mathcal{A}_t/\hspace{-0.5mm}\sim\rightarrow \mathfrak{X}_t,[A(z)]\mapsto (f,g),
\end{equation*}
which is a rational inverse of \eqref{eq:parametrisationmap}.
For example, consider the seventh case, i.e. the curve defined by $u_{10}=u_{11}=0$ in $\mathcal{A}_t/\hspace{-0.5mm}\sim$. By the explicit formulas \eqref{eq:rational_inverse}, clearly $(f,g)=(\infty,\kappa_\infty^{-1})$ on this curve. In terms of the local coordinates $(u_7,v_7)$ around $E_7$, we find
\begin{equation*}
u_7=0,\quad
v_7=d_1+\kappa_\infty^{-1}(\kappa_1+\kappa_1^{-1}),
\end{equation*}
on the curve defined by $u_{10}=u_{11}=0$ and thus the map is regular around this curve.
It follows that the mapping \eqref{eq:parametrisationmap} is an isomorphism.
\end{proof}

It now remains to consider $E_8$.
\begin{remark}\label{rem:special}
We observe that, for $q\neq 0$, the rational mapping
\begin{equation}\label{eq:parametrisation_initial}
 \mathcal{X}_t \rightarrow \mathcal{A}_t/\sim,\, (f,g)\mapsto [A(z)],
\end{equation}
is singular along $E_8$. From a Riemann-Hilbert point of view, times $t$ where $(f,g)$ take value in $E_8$ are exactly those times for which the corresponding Riemann-Hilbert problem has no solution, see \cite[Theorem 2.12]{jr_qp6}.
\end{remark}



\subsection{Crystal limit of canonical solutions}\label{sec:canonical_sol}
In this section, we study canonical solutions of the linear system \eqref{eq:linear_problem} and their analyticity as functions of $q$ in a domain containing $q=0$.
Before we do so, we recall some standard notation and results.

Firstly,  the $q$-Pochhammer symbol
\begin{equation*}
(z;q)_\infty=\prod_{k=0}^{\infty}{(1-q^kz)}.
\end{equation*}
Here, the right-hand side is locally uniformly convergent in $(z,q)\in\mathbb{C}\times \mathbb{D}$, where $\mathbb{D}$ denotes the open unit disc, $\mathbb{D}:=\{q\in\mathbb{C}: |q|<1\}$.
Secondly, we will use the $\mathit{q}$-theta function, defined by
\begin{equation*}
\theta_q(z)=(q;q)_\infty(z;q)_\infty(q/z;q)_\infty.
\end{equation*}
This function is analytic in $(z,q)\in\mathbb{C}^*\times \mathbb{D}$, and admits the following convergent expansion in its domain,
\begin{equation}\label{eq:jacobitripleproduct}
 \theta_q(z)=\sum_{n=-\infty}^\infty(-1)^n q^{\frac{1}{2}n(n-1)}z^n.
\end{equation}
This is known as the Jacobi triple product formula. By collecting terms with similar powers of $q$, \eqref{eq:jacobitripleproduct} can be rewritten as a convergent power series in $q$ around $q=0$, which in particular gives
\begin{equation}\label{eq:theta_crystal}
    \theta_q(z)=1-z+\mathcal{O}(q),
\end{equation}
as $q\rightarrow 0$, locally uniformly in $z\in\mathbb{C}^*$.

For $n\in\mathbb{N}^*$, we use the common abbreviation for repeated products of this function,
\begin{equation*}
\theta_q(z_1,\ldots,z_n)=\theta_q(z_1)\cdot \ldots\cdot \theta_q(z_n).
\end{equation*}

As shown by Carmichael \cite{carmichael1912}, under the respective conditions $\kappa_0^2\notin q^{\mathbb{Z}}$ and $\kappa_\infty^2\notin q^{\mathbb{Z}}$, the linear system \eqref{eq:linear_problem} has solutions $Y_\infty(z)$ and  $Y_0(z)$
of the form
\begin{equation}\label{eq:linear_sys_solutions}
\begin{aligned}
Y_\infty(z)&=z^{\log_q(z)-1}\Psi_\infty(z) \,z^{-\log_q(\kappa_\infty)\sigma_3},\\
Y_0(z)&=z^{\log_q(t)}\Psi_0(z)\, z^{\log_q(\kappa_0)\sigma_3},
\end{aligned}
\end{equation}
where $\Psi_\infty(z)$ is the unique analytic matrix function on $ \mathbb{P}^1\setminus\{0\}$ satisfying
\begin{equation}\label{eq:solinfty}
\begin{aligned}
\Psi_\infty(z)&=\frac{q^2}{z^2} \,A(z/q) \,\Psi_\infty(z/q)\,\kappa_\infty^{-\sigma_3},\\
\Psi_\infty(\infty)&=I,
\end{aligned}
\end{equation}
and $\Psi_0(z)^{-1}$ is the unique analytic matrix function on $\mathbb{C}$ that satisfies
\begin{equation}\label{eq:solzero}
\begin{aligned}
\Psi_0(z)^{-1}&=t^{-1}\kappa_0^{-\sigma_3} \Psi_0(qz)^{-1} A(z),\\
\Psi_0(0)^{-1}&=H^{-1}.
\end{aligned}
\end{equation}

We now consider these two solutions around $q=0$. Their existence is predicated on the conditions $\kappa_0^2,\kappa_\infty^2\notin q^{\mathbb{Z}}$. For sufficiently small $q$, these conditions are trivially satisfied and we have the following result regarding the analyticity of the matrix functions above as functions of $q$, where we use the notation $D_R:=\{q\in\mathbb{C}:|q|<R\}$.

\begin{proposition}\label{prop:crystal_canonical_sol}
Let
$$R_0=\min(|\kappa_0|^2,|\kappa_0|^{-2}),$$ 
then $\Psi_0(z)^{-1}$ is analytic in $(z,q)\in \mathbb{C}\times D_{R_0}$ and
\begin{align*}
       \Psi_0(z)^{-1}&=t^{-1}\kappa_0^{-\sigma_3}H^{-1}A(z)+\mathcal{O}(q),\\
       \Psi_0(qz)^{-1}&=H^{-1}+\mathcal{O}(q),
\end{align*}
locally uniformly in $z\in\mathbb{C}$, as $q\rightarrow 0$, where $H$ is defined in Equation \eqref{eq:diagonal}.

Similarly, set $$R_\infty=\min(|\kappa_\infty|^2,|\kappa_\infty|^{-2}),$$  then $\Psi_\infty(z)$ is analytic in $(z,q)\in (\mathbb{P}^1\setminus \{0\})\times D_{R_\infty}$ and
\begin{align}
   \Psi_\infty(z)&=I+\mathcal{O}(q),\label{eq:psiinfq0a}\\
   \Psi_\infty(qz)&=z^{-2}A(z)\kappa_\infty^{-\sigma_3}+\mathcal{O}(q), \label{eq:psiinfq0b}
\end{align}
locally uniformly in $z\in\mathbb{P}^1\setminus \{0\}$,  as $q\rightarrow 0$.
\end{proposition}
\begin{proof}
We consider the second case (the proof of the first assertion is similar). 
Note that the condition $\kappa_\infty^2\notin q^{\mathbb{Z}}$ for its existence is trivially satisfied for $|q|<R_\infty$. Fix any $0<R<R_\infty$, then $\Psi_\infty(z)$ is well-defined for all $q$ inside the punctured disc $\{0<|q|\leq R\}$.


By Carmichael's construction, the matrix function $\Psi_\infty(z)$ is analytic in $z\in\mathbb{P}^1\setminus\{0\}$, which implies that it has a power series expansion around $z=\infty$,
\begin{equation}\label{eq:psiinfseries}
  \Psi_\infty(z)=I+\sum_{n=1}^\infty z^{-n}V_n(q),
\end{equation}
that converges in that domain.  By estimating the coefficients in this expansion, we will derive that $\Psi_\infty(z)$ is analytic in $q$ on $D_{R}$.

From equations \eqref{eq:solinfty}, we find the following recursive formula for the coefficients in \eqref{eq:psiinfseries},
\begin{equation}\label{eq:recurrence}
    q^{-n} V_n(q) \kappa_\infty^{\sigma_3}-\kappa_\infty^{\sigma_3} V_n(q)=q A_1 V_{n-1}(q)+q^2 A_0 V_{n-2}(q),
\end{equation}
 for $n\geq 1$, with initial conditions
\begin{equation*}
    V_0(q)=I,\quad V_{-1}(q)=0.
\end{equation*}

From the recurrence, we immediately see that each coefficient matrix $V_n(q)$ is a rational function in $q$. Multiplying Equation \eqref{eq:recurrence} by $q^n$, we see that $V_n(q)$ is analytic on $D_R$ with
\begin{equation}\label{eq:estimate_V}
    V_n(q)=\mathcal{O}(q^{\lfloor\frac{1}{4}n(n+6)\rfloor}),
\end{equation}
as $q\rightarrow 0$.

Now, choose an $\alpha\geq 1$ such that
\begin{equation*}
    \lVert A_1\rVert_{\max}\leq \alpha,\quad \lVert A_0\rVert_{\max}\leq \alpha^2, \quad  
\end{equation*}
where $\lVert\cdot\rVert_{\max}$ denotes the max-norm on $2\times 2$ matrices. Further, write $k=|\kappa_\infty|$ and define
\begin{equation*}
    \beta=\min\{(1-R)k,
    (1-R)k^{-1},|k-R\,k^{-1}|,
    |k^{-1}-R\,k|
    \},
\end{equation*}
so that $0<\beta<1$ and, for all $n\geq 1$,
\begin{equation*}
    \beta\leq \min{(|(q^{n}-1)\kappa_\infty|,|(q^{n}\kappa_\infty^{-1}-\kappa_\infty)|,|q^{n}\kappa_\infty-\kappa_\infty^{-1}|,|(q^{n}-1)\kappa_\infty^{-1}|)}.
\end{equation*}

Then, for $0<|q|\leq R$ and $n\geq 1$, the recurrence relation in equation \eqref{eq:recurrence} implies
\begin{align*}
 \lVert V_n(q)\rVert_{\max}&\leq \frac{\lVert q A_1 V_{n-1}(q)+q^2 A_0 V_{n-2}(q)\rVert_{\max}}{\min{(|(q^{-n}-1)\kappa_\infty|,|(q^{-n}\kappa_\infty^{-1}-\kappa_\infty)|,|q^{-n}\kappa_\infty-\kappa_\infty^{-1}|,|(q^{-n}-1)\kappa_\infty^{-1}|)}}\\
&\leq \frac{R^n}{\beta}\lVert q A_1 V_{n-1}(q)+q^2 A_0 V_{n-2}(q)\rVert_{\max}\\
&\leq R^{n} \frac{\alpha R}{\beta}\lVert V_{n-1}(q)\rVert_{\max}+
R^{n} \frac{(\alpha R)^2}{\beta}\lVert V_{n-2}(q)\rVert_{\max}\\
&\leq R^{n} \frac{\alpha R}{\beta}\lVert V_{n-1}(q)\rVert_{\max}+
R^{n} \bigg(\frac{\alpha R}{\beta}\bigg)^2\lVert V_{n-2}(q)\rVert_{\max}.
\end{align*}
Here, the second inequality follows from multiplication of numerator and denominator by $|q|^n$.
Therefore, by induction, we obtain the estimate
\begin{equation*}
    \lVert V_n(q)\rVert_{\max}\leq \bigg(\frac{2\alpha}{\beta}\bigg)^n R^{\frac{1}{4}n(n+6)},
\end{equation*}
for all $|q|\leq R$, $n\geq 0$.

This means that the series representation of $\Psi_\infty(z)$ in equation \eqref{eq:psiinfseries}
is uniformly absolutely convergent on
\begin{equation*}
    \{(z,q)\in\mathbb{P}^1\times \mathbb{C}:|z|\geq \epsilon, |q|\leq R\},
\end{equation*}
for any $\epsilon>0$. In particular, $\Psi_\infty(z)$ is analytic in
$(z,q)\in (\mathbb{P}^1\setminus\{0\})\times D_R$. Since this holds for any $0<R<R_\infty$, it also holds when $R=R_\infty$. Furthermore, it follows directly from estimate \eqref{eq:estimate_V} that $V_n(q)$ vanishes faster than $\mathcal O(q)$ as $q\rightarrow 0$ for all $n\geq 1$. As a consequence, we obtain estimate \eqref{eq:psiinfq0a}.

Next, from the $q$-difference equation in \eqref{eq:solinfty}, we obtain
\begin{align*}
\Psi_\infty(q\,z)&=\frac{1}{z^2} A(z) \Psi_\infty(z)\kappa_\infty^{-\sigma_3}\\
&=\frac{1}{z^2} A(z) (I+\mathcal{O}(q))\kappa_\infty^{-\sigma_3}\\
&=\frac{1}{z^2} A(z)\kappa_\infty^{-\sigma_3}+\mathcal{O}(q),
\end{align*}
as $q\rightarrow 0$, proving equation \eqref{eq:psiinfq0b}, as desired.
\end{proof}

From Proposition \ref{prop:crystal_canonical_sol}, we obtain the crystal limits of the matrix functions $\Psi_0(z)$ and $\Psi_\infty(z)$,
\begin{align*}
    \Psi_0(z)&\xrightarrow{q\rightarrow 0} \Psi_0^\diamond(z),  & \Psi_0^\diamond(z)&=t\,A(z)^{-1}H\kappa_0^{\sigma_3},\\
    \Psi_\infty(z)&\xrightarrow{q\rightarrow 0} \Psi_\infty^\diamond(z),  & \Psi_\infty^\diamond(z)&=I. 
\end{align*}

\subsection{Crystal limit of the connection matrix}\label{sec:connection_matrix}

In this section, we study the {\em connection matrix} relating the two matrix functions $\Psi_0(z)$ and $\Psi_\infty(z)$,
\begin{equation*}
    C(z):=\Psi_0(z)^{-1}\Psi_\infty(z).
\end{equation*}
Recall from \cite{jr_qp6}, that
for $q\neq 0$, this matrix has the following analytic properties with respect to $z$.
\begin{enumerate}
    \item It is a single-valued analytic function in $z\in\mathbb{C}^*$.
    \item It satisfies the $q$-difference equation
    \begin{equation*}
        C(q\,z)=\frac{t}{z^2}\kappa_0^{\sigma_3}C(z)\kappa_\infty^{-\sigma_3}.
    \end{equation*}
    \item Its determinant is given by
    \begin{equation*}
    |C(z)|=c\, \theta_q\left(\kappa_t^{+1}\frac{z}{t},\kappa_t^{-1}\frac{z}{t},\kappa_1^{+1}z,\kappa_1^{-1}z\right),
    \end{equation*}
    for some $c\in\mathbb{C}^*$.
\end{enumerate}
Following \cite[Definition 2.3]{jr_qp6}, we introduce a corresponding monodromy manifold.
\begin{definition}\label{def:monodromy_manifold}
For $q\neq 0$, we define $\mathcal{M}_t$ to be the space of connection matrices satisfying properties (1),(2) and (3) above, quotiented by arbitrary left and right-multiplication by invertible diagonal matrices.
We refer to $\mathcal{M}_t$ as the monodromy manifold of $q\Psix$.
\end{definition}

Using Proposition \ref{prop:crystal_canonical_sol}, we can compute the crystal limit of the connection matrix,
\begin{align}
    C(z)&=\Psi_0(z)^{-1}\Psi_\infty(z)\nonumber\\
    &=(t^{-1}\kappa_0^{-\sigma_3}H^{-1}A(z)+\mathcal{O}(q))(I+\mathcal{O}(q))\label{eq:connectioncrytallim}\\
    &=t^{-1}\kappa_0^{-\sigma_3}H^{-1}A(z)+\mathcal{O}(q)\nonumber,
\end{align}
as $q\rightarrow 0$, which holds locally uniformly in $z\in\mathbb{C}^*$.
So, we find that
\begin{equation}\label{eq:connectioncrytallimat0}
  C(z)\xrightarrow{q\rightarrow 0} C^\diamond(z),\qquad   C^\diamond(z)=t^{-1}\kappa_0^{-\sigma_3}H^{-1}A(z).
\end{equation}
This matrix function has the following characterising properties.
\begin{itemize}
    \item[(1)'] The matrix $C^\diamond(z)$ is a degree two matrix polynomial, 
    \begin{equation*}
        C^\diamond(z)=C_0^\diamond+z\, C_1^\diamond+z^2\, C_2^\diamond.
    \end{equation*}
     \item[(2)'] The constant and leading order coefficient of $C^\diamond(z)$ are related by
     \begin{equation*}
       C_0^\diamond=t\, \kappa_0^{\sigma_3}\, C_2^\diamond\,\kappa_\infty^{-\sigma_3}.
     \end{equation*}
     \item[(3)'] Its determinant is given by
     \begin{equation*}
    |C^\diamond(z)|=c\, (z-\kappa_t^{+1}t)(z-\kappa_t^{-1}t)(z-\kappa_1^{+1})(z-\kappa_1^{-1}),
    \end{equation*}
    for some $c\in\mathbb{C}^*$.
\end{itemize}
In analogy with Definition \ref{def:monodromy_manifold}, we thus make the following definition.
\begin{definition}\label{def:monodromy_manifold_crystal}
We define $\mathcal{M}_t^\diamond$ to be the space of connection matrices satisfying properties (1)',(2)' and (3)' above, quotiented by arbitrary left and right-multiplication by invertible diagonal matrices.
We refer to $\mathcal{M}_t^\diamond$ as the crystal limit of the monodromy manifold of $q\Psix$.
\end{definition}

%% file: monodromy.tex
\section{A Segre surface and the Riemann-Hilbert correspondence}\label{s:segre_RHP}
In this section, we define a precise mapping that is an instance of the Riemann-Hilbert correspondence in the setting of $q\Psix$. The domain of this mapping is the initial value space $\mathcal{X}_t$, constructed in Definition \ref{def:initial_space}. The co-domain is an affine Segre surface. We then study this mapping, and its co-domain, in the crystal limit $q\rightarrow 0$. This leads in particular to our main result, Theorem \ref{thm:main}.


\subsection{Tyurin parameters}\label{subsec:tyurin}
In this section, we study the crystal limit of the Tyurin parameters, introduced in \cite[Section 2.4]{jr_qp6}, associated with the monodromy manifold in Definition \ref{def:monodromy_manifold}.

For any $2\times2$ matrix $R$ of rank one, let $R_1$ and $R_2$ be respectively its first and second column, then we define $\pi(R)\in\mathbb{P}^1$ by
\begin{equation*}
R_1=\pi(R)R_2.
\end{equation*}
 Denoting
\begin{equation}\label{eq:intro_xnotation}
    (x_1,x_2,x_3,x_4)=(\kappa_t t,\kappa_t^{-1} t,\kappa_1 ,\kappa_1^{-1}),
\end{equation}
the Tyurin parameters are defined by
\begin{equation}\label{eq:tyurinparam}
\rho_k=\pi(C(x_k)),\quad (1\leq k\leq 4).
\end{equation}
The Tyurin parameters $\rho=(\rho_1,\rho_2,\rho_3,\rho_4)$ are invariant under left multiplication of $C(z)$ by diagonal matrices. However, multiplication by diagonal matrices from the right has the effect of scaling $\rho\mapsto c\, \rho$, for some $c\in\mathbb{C}^*$.
Therefore, the Tyurin parameters $\rho$ naturally lie in $(\mathbb{P}^1)^4/\mathbb{C}^*$.


In \cite{jr_qp6}, it is shown that the Tyurin parameters satisfy the following homogeneous multilinear equation, written in inhomogeneous coordinates,
\begin{align}\label{eq:tyurineq}
    &T(\rho)=0,\\
&T(\rho):=T_{12}\,\rho_1\rho_2+T_{13}\,\rho_1\rho_3+T_{14}\,\rho_1\rho_4+T_{23}\,\rho_2\rho_3+T_{24}\,\rho_2\rho_4+T_{34}\,\rho_3\rho_4,\nonumber
\end{align}
with coefficients given by
\begin{align*}
T_{12}&= \theta_q\left(\kappa_t^2,\kappa_1^2\right)\theta_q\left(\kappa_0\kappa_\infty^{-1}t,\kappa_0^{-1}\kappa_\infty^{-1}t\right)\kappa_\infty^2,\\
T_{34}&= \theta_q\left(\kappa_t^2,\kappa_1^2\right)\theta_q\left(\kappa_0\kappa_\infty t,\kappa_0^{-1}\kappa_\infty t\right),\\
T_{13}&=- \theta_q\left(\kappa_t\kappa_1^{-1}t,\kappa_t^{-1}\kappa_1t\right)\theta_q\left(\kappa_t\kappa_1\kappa_0^{-1}\kappa_\infty^{-1},\kappa_0\kappa_t\kappa_1\kappa_\infty^{-1}\right)\kappa_\infty^2,\\
T_{24}&=-\theta_q\left(\kappa_t\kappa_1^{-1}t,\kappa_t^{-1}\kappa_1t\right) \theta_q\left(\kappa_0\kappa_t\kappa_1\kappa_\infty,\kappa_t\kappa_1\kappa_\infty\kappa_0^{-1}\right),\\
T_{14}&= \theta_q\left(\kappa_t\kappa_1t,\kappa_t^{-1}\kappa_1^{-1}t\right)\theta_q\left(\kappa_1\kappa_\infty\kappa_0^{-1}\kappa_t^{-1},\kappa_0\kappa_1\kappa_\infty\kappa_t^{-1}\right)\kappa_t^2,\\
T_{23}&= \theta_q\left(\kappa_t\kappa_1t,\kappa_t^{-1}\kappa_1^{-1}t\right)\theta_q\left(\kappa_t\kappa_\infty\kappa_0^{-1}\kappa_1^{-1},\kappa_0\kappa_t\kappa_\infty\kappa_1^{-1}\right)\kappa_1^2.
\end{align*}
In homogeneous coordinates $\rho_k=[\rho_k^x: \rho_k^y]\in \mathbb{ P}^1$, $1\le k\le 4$, equation \eqref{eq:tyurineq} reads
\begin{align*}
0=\,&T_{12}\rho_1^x\rho_2^x\rho_3^y\rho_4^y+T_{13}\rho_1^x\rho_2^y\rho_3^x\rho_4^y
+T_{14}\rho_1^x\rho_2^y\rho_3^y\rho_4^x+\\
&T_{23}\rho_1^y\rho_2^x\rho_3^x\rho_4^y+T_{24}\rho_1^y\rho_2^x\rho_3^y\rho_4^x+T_{34}\rho_1^y\rho_2^y\rho_3^x\rho_4^x.
\end{align*}

Furthermore, it is shown in \cite{jr_qp6} that the following inequality holds,
\begin{equation}\label{eq:tyurinnonzero}
   \widehat{T}(\rho)\neq 0,
\end{equation}
where $\widehat{T}(\rho)$ is the polynomial obtained by setting $\kappa_0= 1$ in $T(\rho)$.
In other words, if we denote $\widehat{T}_{ij}=T_{ij}|_{\kappa_0= 1}$, then
\begin{align*}
0\neq\,&\widehat{T}_{12}\rho_1^x\rho_2^x\rho_3^y\rho_4^y+\widehat{T}_{13}\rho_1^x\rho_2^y\rho_3^x\rho_4^y
+\widehat{T}_{14}\rho_1^x\rho_2^y\rho_3^y\rho_4^x+\\
&\widehat{T}_{23}\rho_1^y\rho_2^x\rho_3^x\rho_4^y+\widehat{T}_{24}\rho_1^y\rho_2^x\rho_3^y\rho_4^x+\widehat{T}_{34}\rho_1^y\rho_2^y\rho_3^x\rho_4^x.
\end{align*}
Equations \eqref{eq:tyurineq} and \eqref{eq:tyurinnonzero}
completely describe the possible values of the Tyurin parameters \cite[Theorem 2.15]{jr_qp6}.

Next, we consider the crystal limit of the Tyurin parameters. Due to equations \eqref{eq:connectioncrytallim} and \eqref{eq:connectioncrytallimat0}, they remain well-defined in the limit as $q\rightarrow 0$, and we have
\begin{equation*}
    \rho_k\xrightarrow{q\rightarrow 0}\rho_k^\diamond,\qquad 
    \rho_k^\diamond=\pi[A(x_k)],
\end{equation*}
for $1\leq k\leq 4$, where we used that $\pi[\cdot]$ is invariant under left-multiplication by invertible matrices.
Now, equation \eqref{eq:tyurineq} continues to hold, and, due to the limiting behaviour of $\theta_q(\cdot)$ in \eqref{eq:theta_crystal}, its coefficients simplify to rational functions with respect to the parameters $(\kappa,t)$ as $q\rightarrow 0$. Namely
\begin{equation*}
    T_{ij}\xrightarrow{q\rightarrow 0}T_{ij}^\diamond,
\end{equation*}
for $1\leq i<j\leq 4$, where
\begin{equation}\label{eq:Tcoefq0}
\begin{aligned}
T_{12}^\diamond&=(\kappa_t^2-1)(\kappa_1^2-1)(\kappa_0t-\kappa_\infty)(\kappa_0^{-1}t-\kappa_\infty),\\
T_{34}^\diamond&=(\kappa_t^2-1)(\kappa_1^2-1)(\kappa_0\kappa_\infty t-1)(\kappa_0^{-1}\kappa_\infty t-1),\\
T_{13}^\diamond&= -(\kappa_t\kappa_1^{-1}t-1)(\kappa_t^{-1}\kappa_1t-1)(\kappa_0\kappa_t\kappa_1-\kappa_\infty)(\kappa_0^{-1}\kappa_t\kappa_1-\kappa_\infty),\\
T_{24}^\diamond&= -(\kappa_t\kappa_1^{-1}t-1)(\kappa_t^{-1}\kappa_1t-1)(\kappa_0\kappa_t\kappa_1\kappa_\infty-1)(\kappa_0^{-1}\kappa_t\kappa_1\kappa_\infty-1),\\
T_{14}^\diamond&=(\kappa_t\kappa_1t-1)(\kappa_t^{-1}\kappa_1^{-1}t-1)(\kappa_0\kappa_1\kappa_\infty-\kappa_t)(\kappa_0^{-1}\kappa_1\kappa_\infty-\kappa_t),\\
T_{23}^\diamond&=(\kappa_t\kappa_1t-1)(\kappa_t^{-1}\kappa_1^{-1}t-1)(\kappa_0\kappa_t\kappa_\infty-\kappa_1)(\kappa_0^{-1}\kappa_t\kappa_\infty-\kappa_1).
\end{aligned}
\end{equation}
Similarly, the coefficients $\widehat{T}_{ij}$ in inequality \eqref{eq:tyurinnonzero} simplify to rational functions $\widehat{T}_{ij}^\diamond$ at $q=0$ and
$\widehat{T}_{ij}^\diamond=T_{ij}^\diamond|_{\kappa_0= 1}$.
Using the explicit parametrisation of $A$ in terms of $(f,g,w)$, we obtain the following expressions for the Tyurin parameters at $q=0$,
\begin{equation}\label{eq:rhocrystal}
\begin{aligned}
    \rho_1^\diamond=\,&\frac{P_1(f,g)}{w f g(\kappa_\infty^2-1)}, &&&     \rho_3^\diamond=\,&\frac{P_3(f,g)}{w f g(\kappa_\infty^2-1)(f-\kappa_1)},\\
    \rho_2^\diamond=\,&\frac{P_2(f,g)}{w f g(\kappa_\infty^2-1)},
&&&
    \rho_4^\diamond=\,&\frac{P_4(f,g)}{w f g(\kappa_\infty^2-1)(f-\kappa_1^{-1})},
\end{aligned}
\end{equation}
where
\begin{equation}\label{eq:Pdefi}
\begin{aligned}
    P_1(f,g)=\,&\kappa_\infty^3(g-\kappa_0 t)(g-\kappa_0^{-1}t)
    +\kappa_\infty( \kappa_\infty g-1)^2f^2\\
    &-\kappa_\infty^2 f(g-1/\kappa_\infty)(\kappa_\infty(\kappa_1+\kappa_1^{-1})g-t(\kappa_\infty^2 \kappa_t+\kappa_t^{-1})),\\
    P_2(f,g)=\,&\kappa_\infty^3(g-\kappa_0 t)(g-\kappa_0^{-1}t)+\kappa_\infty( \kappa_\infty g-1)^2f^2,\\
    &-
    \kappa_\infty^2 f(g-1/\kappa_\infty)(\kappa_\infty(\kappa_1+\kappa_1^{-1})g-t(\kappa_\infty^2 \kappa_t^{-1}+\kappa_t))\\    
    P_3(f,g)=\,&\kappa_\infty^3g^2(f-\kappa_1)^2(f-\kappa_1^{-1})-\kappa_\infty^2g(f-\kappa_1)Q(f,\kappa_1)\\
    &+\kappa_\infty(f-\kappa_1 \kappa_\infty)(f-\kappa_t t)(f-\kappa_t^{-1} t),\\
    P_4(f,g)=\,&\kappa_\infty^3g^2(f-\kappa_1^{-1})^2(f-\kappa_1)-\kappa_\infty^2g(f-\kappa_1^{-1})Q(f,\kappa_1^{-1})\\
    &+\kappa_\infty(f-\kappa_1^{-1} \kappa_\infty)(f-\kappa_t t)(f-\kappa_t^{-1} t),
\end{aligned}
\end{equation}
with
\begin{equation*}
    Q(f,\kappa_1):=(f-\kappa_t t)(f-\kappa_t^{-1} t)+(f-\kappa_1^{-1})(f-\kappa_1 \kappa_\infty^2)-(t-\kappa_0 \kappa_\infty)(t-\kappa_0^{-1}\kappa_\infty).
\end{equation*}
One can check directly that equation \eqref{eq:tyurineq} continues to hold at $q=0$, by substitution of formulas \eqref{eq:rhocrystal} for the Tyurin parameters.

Note that the auxiliary variable $w$ in the linear system traces out an orbit in $(\mathbb{P}^1)^4$, with respect to scalar multiplication, when varied in $\mathbb{C}^*$.
Thus, considering the Tyurin parameters as elements of $(\mathbb{P}^1)^4/\mathbb{C}^*$, we obtain a mapping from the initial value space to this quotient space,
\begin{equation}\label{eq:coordinate_mapping_init}
    \mathcal{X}_t\rightarrow (\mathbb{P}^1)^4/\mathbb{C}^*,\,(f,g)\mapsto [\rho],
\end{equation}
for nonzero $q$ and for $q=0$.

Whilst for nonzero $q$, this mapping is expected to be higher-order transcendental, it becomes rational at $q=0$, as manifest by the explicit formulas above for $\rho^\diamond$.

There is another interesting phenomenon happening at $q=0$. For $q\neq 0$, inequality \eqref{eq:tyurinnonzero} holds for any point in the initial value space $\mathcal{X}_t$. However, at $q=0$, this inequality is violated on the exceptional curve $E_8$.
To see this, we use the local coordinates $(u_8,v_8)$ in \eqref{eq:exceptional_para}. Direct substitution into the formulas for $\rho^\diamond$ above gives
\begin{equation*}
    \rho_k^\diamond=\frac{\kappa_\infty}{(\kappa_\infty^2-1)\,v_8\, w\, u_8^2}+\mathcal{O}(u_8^{-1})\qquad (u_8\rightarrow 0),
\end{equation*}
for $1\leq k\leq 4$. This means that the mapping \eqref{eq:coordinate_mapping_init} sends $E_8$ onto the single point 
\begin{equation}\label{eq:rho_sing}
    [(1,1,1,1)]\in(\mathbb{P}^1)^4/\mathbb{C}^*.
\end{equation}
This point satisfies equality \eqref{eq:tyurineq} with $q=0$, for any value of $\kappa_0$, including $\kappa_0=1$, and thus violates inequality \eqref{eq:tyurinnonzero}. We further note that inequality \eqref{eq:tyurinnonzero} is satisfied away from $E_8$.


\subsection{An affine Segre surface and the Riemann-Hilbert correspondence}\label{subsec:segre}
In this section, we recall the construction of an affine variety in \cite[Section 2.4]{jr_qp6} naturally associated with the monodromy manifold of $q\Psix$. This construction relies on equality \eqref{eq:tyurineq} and inequality \eqref{eq:tyurinnonzero} satisfied by the Tyurin parameters.

Take $1\leq i<j\leq 4$ and define
\begin{align}
     \eta_{ij}:&=\frac{T_{ij}\,\rho_i\rho_j}{\theta_q(\kappa_0,\kappa_0^{-1})\widehat{T}(\rho)}\label{eq:eta_defi}\\
     &=\frac{1}{\theta_q(\kappa_0,\kappa_0^{-1})}\frac{T_{ij}\,\rho_i^{x}\rho_j^{x}\rho_k^{y}\rho_l^{y}}{\widehat{T}_{12}\,\rho_1^{x}\rho_2^{x}\rho_3^{y}\rho_4^{y}+\widehat{T}_{13}\,\rho_1^{x}\rho_2^{y}\rho_3^{x}\rho_4^{y}+\ldots+\widehat{T}_{34}\,\rho_1^{y}\rho_2^{y}\rho_3^{x}\rho_4^{x}}\nonumber
\end{align}
where $k,l$ are such that $\{i,j,k,l\}=\{1,2,3,4\}$.

The $\eta_{ij}$, $1\leq i<j\leq 4$, are six well-defined global coordinates on the monodromy manifold $\mathcal{M}_t$, defined in Definition \ref{def:monodromy_manifold}.
They satisfy the following four equations,
\begin{subequations}\label{eq:eta_equations}
\begin{align}
    &\eta_{12}+\eta_{13}+\eta_{14}+\eta_{23}+\eta_{24}+\eta_{34}=0,\label{eq:eta_equationsa}\\
    &a_{12}\,\eta_{12}+a_{13}\,\eta_{13}+a_{14}\,\eta_{14}+a_{23}\,\eta_{23}+a_{24}\,\eta_{24}+a_{34}\,\eta_{34}=1,\label{eq:eta_equationsb}\\
    &\eta_{13}\,\eta_{24}-b_1\,\eta_{12}\,\eta_{34}=0,\label{eq:eta_equationsc}\\ 
    &\eta_{14}\,\eta_{23}-b_2\,\eta_{12}\,\eta_{34}=0,\label{eq:eta_equationsd}
\end{align}
\end{subequations}
where the coefficients $a_{ij}=\widehat{T}_{ij}/T_{ij}$, $1\leq i<j\leq 4$, read
\begin{align*}
    &a_{12}=\prod_{\epsilon=\pm 1}\frac{\theta_q\big(\kappa_0^\epsilon\big)\theta_q\big(\kappa_\infty^{-1}t\big)}{\theta_q\big(\kappa_0^{\epsilon}\kappa_\infty^{-1}t\big)}, &
    &a_{34}=\prod_{\epsilon=\pm 1}\frac{\theta_q\big(\kappa_0^\epsilon\big)\theta_q\big(\kappa_\infty t\big)}{\theta_q\big(\kappa_0^{\epsilon}\kappa_\infty t\big)},\\
    &a_{13}=\prod_{\epsilon=\pm 1}\frac{\theta_q\big(\kappa_0^\epsilon\big)\theta_q\big(\kappa_t\kappa_1\kappa_\infty^{-1}\big)}{\theta_q\big(\kappa_0^{\epsilon}\kappa_t\kappa_1\kappa_\infty^{-1}\big)}, &
    &a_{24}=\prod_{\epsilon=\pm 1}\frac{\theta_q\big(\kappa_0^\epsilon\big)\theta_q\big(\kappa_t\kappa_1\kappa_\infty \big)}{\theta_q\big(\kappa_0^{\epsilon}\kappa_t\kappa_1\kappa_\infty \big)},\\
    &a_{14}=\prod_{\epsilon=\pm 1}\frac{\theta_q\big(\kappa_0^\epsilon\big)\theta_q\big(\kappa_t^{-1}\kappa_1\kappa_\infty\big)}{\theta_q\big(\kappa_0^{\epsilon}\kappa_t^{-1}\kappa_1\kappa_\infty\big)}, &
    &a_{23}=\prod_{\epsilon=\pm 1}\frac{\theta_q\big(\kappa_0^\epsilon\big)\theta_q\big(\kappa_t\kappa_1^{-1}\kappa_\infty \big)}{\theta_q\big(\kappa_0^{\epsilon}\kappa_t\kappa_1^{-1}\kappa_\infty \big)},\\
\end{align*}
and
\begin{equation}\label{eq:b1b2defi}
    b_1=\frac{T_{13}T_{24}}{T_{12}T_{34}},\qquad b_2=\frac{T_{14}T_{23}}{T_{12}T_{34}}.
\end{equation}

\begin{definition}\label{def:affine_variety}
We denote by $\mathcal{F}_t$ the affine Segre surface in
\begin{equation*}
    \{\eta=(\eta_{12},\eta_{13},\eta_{14},\eta_{23},\eta_{24},\eta_{34})\in\mathbb{C}^6\}
\end{equation*}
defined by equations \eqref{eq:eta_equations}.
\end{definition}


We now have all the ingredients to define the Riemann-Hilbert correspondence in this context.
\begin{definition}\label{def:rh}
For $q\neq 0$, we define the mapping
\begin{equation*}
   \mathrm{RH}_t:\mathcal{X}_t\rightarrow \mathcal{F}_t,(f,g)\mapsto \eta,
\end{equation*}
which associates to any $(f,g)$, not on the exceptional curve $E_8$, the $\eta$-coordinates of its corresponding connection matrix $C(z)$ via the linear problem \eqref{eq:linear_problem}. If $(f,g)$ lies on the exceptional curve $E_8$, then $(\overline{f},\overline{g})\in \mathcal{X}_{qt}$ does not lie on $E_8$ and we define
\begin{equation}\label{eq:RHinvariance}
    \mathrm{RH}_t(f,g)=\mathrm{RH}_{qt}(\overline{f},\overline{g}),
\end{equation}
where we note that $\mathcal{F}_{qt}= \mathcal{F}_t$.
\end{definition}

We recall the following important facts concerning the mapping $\mathrm{RH}_t$, see \cite[Proposition 2.6]{rofqpviasymptotics}.
\begin{proposition} \label{prop:rh}
Let $q\neq 0$ and assume the non-resonance conditions \eqref{eq:non_res} and non-splitting conditions \eqref{eq:intro_irreducibleparameter}. Then, the mapping $\mathrm{RH}_t$ is a bi-holomorphism. It further commutes with the $q\Psix$ time-evolution, in the sense that equation \eqref{eq:RHinvariance} holds for all $(f,g)\in \mathcal{X}_t$.
\end{proposition}

\subsection{Crystal limit of Segre surface and Riemann-Hilbert correspondence}\label{subsec:segre_crystal}
The construction of the Segre surface $\mathcal{F}_t$ remains completely well-defined when we set $q=0$. Due to \eqref{eq:theta_crystal}, its coefficients simplify to rational functions in $(\kappa,t)$ as $q\rightarrow 0$,
\begin{equation*}
    a_{ij}\xrightarrow{q\rightarrow 0} a_{ij}^\diamond, \qquad b_{i}\xrightarrow{q\rightarrow 0} b_{i}^\diamond,
\end{equation*}
where
\begin{align*}
    &a_{12}^\diamond=\frac{(\kappa_0-1)^2(t-\kappa_\infty)^2}{(\kappa_0t-\kappa_\infty)(\kappa_0\kappa_\infty-t)}
    &
    &a_{34}^\diamond=\frac{(\kappa_0-1)^2(\kappa_\infty t-1)^2}{(\kappa_0\kappa_\infty t-1)(\kappa_0-\kappa_\infty t)},\\
    &a_{13}^\diamond=\frac{(\kappa_0-1)^2(\kappa_\infty-\kappa_t\kappa_1)^2}{(\kappa_0\kappa_t\kappa_1-\kappa_\infty)(\kappa_0\kappa_\infty-\kappa_t\kappa_1)},
    &
    &a_{24}^\diamond=\frac{(\kappa_0-1)^2(\kappa_t\kappa_1\kappa_\infty-1)^2}{(\kappa_0\kappa_t\kappa_1\kappa_\infty-1)(\kappa_0-\kappa_t\kappa_1\kappa_\infty)},\\
    &a_{14}^\diamond=\frac{(\kappa_0-1)^2(\kappa_t-\kappa_1\kappa_\infty)^2}{(\kappa_0\kappa_1\kappa_\infty-\kappa_t)(\kappa_0\kappa_t-\kappa_1\kappa_\infty)}, 
    &
    &a_{23}^\diamond=\frac{(\kappa_0-1)^2(\kappa_1-\kappa_t\kappa_\infty)^2}{(\kappa_0\kappa_t\kappa_\infty-\kappa_1)(\kappa_0\kappa_1-\kappa_t\kappa_\infty)},
\end{align*}
and
\begin{align*}
    b_1^\diamond=&\frac{( \kappa_t t-\kappa_1)^2(\kappa_1 t-\kappa_t)^2}{\kappa_t^2\kappa_1^2(\kappa_t^2-1)^2(\kappa_1^2-1)^2}\\
    &\times \frac{(\kappa_0 \kappa_\infty-\kappa_t \kappa_1)(\kappa_0\kappa_t\kappa_1-\kappa_\infty)(\kappa_0-\kappa_t\kappa_1\kappa_\infty)(\kappa_0\kappa_t\kappa_1\kappa_\infty-1)}{(\kappa_0 t-\kappa_\infty)(\kappa_0-\kappa_\infty t)(\kappa_0\kappa_\infty-t)(\kappa_0\kappa_\infty t-1)},\\
    b_2^\diamond=&\frac{(  t-\kappa_t\kappa_1)^2(\kappa_t\kappa_1 t-1)^2}{\kappa_t^2\kappa_1^2(\kappa_t^2-1)^2(\kappa_1^2-1)^2}\\
    &\times \frac{(\kappa_0 \kappa_t-\kappa_1 \kappa_\infty)(\kappa_0\kappa_1\kappa_\infty-\kappa_t)(\kappa_0\kappa_1-\kappa_t\kappa_\infty)(\kappa_0\kappa_t\kappa_\infty-\kappa_1)}{(\kappa_0 t-\kappa_\infty)(\kappa_0-\kappa_\infty t)(\kappa_0\kappa_\infty-t)(\kappa_0\kappa_\infty t-1)}.
\end{align*}
We denote the limiting Segre surface, with coefficients as above, by $\mathcal{F}_t^{\diamond}$.

In \cite{jmr_segre}, it is shown that the Segre surface $\mathcal{F}_t$, with $q\neq 0$, is a completely generic embedded affine Segre surface, for generic values of the parameters.  In particular, its projective completion $\overline{\mathcal{F}}_t\subseteq \mathbb{P}^6$, defined through projective coordinates
 \begin{equation*}
\left[N_\infty:N_{12}:N_{13}:N_{14}:N_{23}:N_{24}:N_{34}\right]=
\left[1:\eta_{12}:\eta_{13}:\eta_{14}:\eta_{23}:\eta_{24}:\eta_{34}\right],
 \end{equation*}
is smooth and the curve at infinity, $\overline{\mathcal{F}}_t\setminus \mathcal{F}_t$, is a smooth irreducible quartic curve, isomorphic to the intersection of two quadric surfaces in $\mathbb{P}^3$, of genus $1$.

 However, under the crystal limit, the embedded affine Segre surface degenerates slightly as the curve at infinity is no longer smooth when $q=0$, as shown in the following proposition.
\begin{proposition}\label{prop:algebraic_characterisation}
The projective completion $\overline{\mathcal{F}}_t^\diamond$ of the limiting Segre surface is smooth and the curve at infinity, $\overline{\mathcal{F}}_t^\diamond\setminus \mathcal{F}_t^\diamond$, is a singular quartic curve, with a singularity at
    \begin{equation}\label{eq:projective}
      N_*=[0:T_{12}^\diamond:T_{13}^\diamond:T_{14}^\diamond:T_{23}^\diamond:T_{24}^\diamond:T_{34}^\diamond].    
    \end{equation}
\end{proposition}
\begin{proof}
Firstly, note that, due to conditions \eqref{eq:non_rescrystal} and \eqref{eq:intro_irreducibleparametercrystal}, the coefficients $b_1^\diamond$ and $b_2^\diamond$ are well-defined and non-zero. 

The surface $\overline{\mathcal{F}}_t^\diamond$ is defined by four equations. To prove its smoothness, we need to show that the corresponding Jacobian does not lose rank at any point on the surface. Without reproducing all the details from the proof of \cite[Proposition 2.6]{rofqpviasymptotics} for $\overline{\mathcal{F}}_t$, we recall that this hinged on the non-vanishing condition
\begin{equation*}
    (b_1-b_2)^2-2(b_1+b_2)+1\neq 0.
\end{equation*}
Under the crystal limit, this becomes
\begin{equation*}
    (b_1^\diamond-b_2^\diamond)^2-2(b_1^\diamond+b_2^\diamond)+1\neq 0.
\end{equation*}
From the definition of these coefficients, the left-hand side equals
\begin{equation*}
\begin{gathered}
   (b_1^\diamond-b_2^\diamond)^2-2(b_1^\diamond+b_2^\diamond)+1=\\
   \frac{\kappa_t^2\,\kappa_1^2(\kappa_0^2-1)^2(\kappa_\infty^2-1)^2}{\kappa_0^2\kappa_\infty^2(\kappa_t^2-1)^2(\kappa_1^2-1)^2}
   \cdot\frac{(t-\kappa_t\kappa_1)(t-\kappa_t^{-1}\kappa_1)(t-\kappa_t\kappa_1^{-1})(t-\kappa_t^{-1}\kappa_1^{-1})}{(t-\kappa_0\kappa_\infty)(t-\kappa_0^{-1}\kappa_\infty)(t-\kappa_0\kappa_\infty^{-1})(t-\kappa_0^{-1}\kappa_\infty^{-1})}.
   \end{gathered}
\end{equation*}
By assumptions \eqref{eq:non_rescrystal} and  \eqref{eq:intro_irreducibleparametercrystal}, this is clearly nonzero and thus $\overline{\mathcal{F}}_t^\diamond$ is smooth.

Using projective coordinates \eqref{eq:projective}, the curve at infinity is given by the hyperplane section $\overline{\mathcal{F}}_t^\diamond\cap\{N_\infty=0\}$ and thus described by
\begin{subequations}
\begin{align}
    &N_{12}+N_{13}+N_{14}+N_{23}+N_{24}+N_{34}=0,\label{eq:hyp1}\\
&a_{12}^\diamond \,N_{12}+a_{13}^\diamond \,N_{13}+a_{14}^\diamond \,N_{14}+a_{23}^\diamond \,N_{23}+a_{24}^\diamond \,N_{24}+a_{34}^\diamond \,N_{34}=0,\label{eq:hyp2}\\
    &N_{13}N_{24}-b_1^\diamond\,N_{12}N_{34}=0,\label{eq:hyp3}\\ 
    &N_{14}N_{23}-b_2^\diamond\,N_{12}N_{34}=0.\label{eq:hyp4}
\end{align}
\end{subequations}
Since the first two equations are linear and the second two quadratic, the curve at infinity is a curve of degree $4$.

The point $N_*$ is the image under the mapping $\rho\mapsto N$, defined through equation \eqref{eq:eta_defi}, of the point \eqref{eq:rho_sing}. It satisfies equations \eqref{eq:hyp3} and \eqref{eq:hyp4}, due the crystal limits of equations \eqref{eq:b1b2defi}. Direct computations shows that it also satisfies \eqref{eq:hyp1} and \eqref{eq:hyp2} and thus defines a point on the curve at infinity.

The Jacobian $J$ of the equations describing the curve at infinity is  given by taking partial derivatives of the left-hand sides of the above four equations, with respect to $(N_{12},N_{13},N_{14},N_{23},N_{24},N_{34})$, i.e., by
\begin{equation*}
    J=\begin{bmatrix}
        1 & 1 & 1 & 1 & 1 & 1 \\
        a_{12}^\diamond & a_{13}^\diamond & a_{14}^\diamond & a_{23}^\diamond & a_{24}^\diamond & a_{34}^\diamond \\
        -b_1^\diamond N_{34} & N_{24} & 0 & 0 & N_{13} & -b_1^\diamond N_{12} \\
        -b_2^\diamond N_{34} & 0 & N_{23} & N_{14} & 0 & -b_2^\diamond N_{12} \\
    \end{bmatrix}.
\end{equation*}
A direct computation shows that the rows of the Jacobian become linearly dependent at $N_*$, explicitly,
\begin{equation*}
    0=u\cdot J|_{N=N_*},
\end{equation*}
where $u=(u_1,u_2,u_3,u_4)$ is given by
\begin{align*}
    u_1&=\frac{1}{(\kappa_0-1)^{2}}, & u_3&=-\frac{\kappa_\infty}{\kappa_0}\frac{(t\kappa_t-\kappa_1)(t\kappa_1-\kappa_t)}{T_{13}^\diamond T_{24}^\diamond},\\
    u_2&=\frac{\kappa_0}{(\kappa_0-1)^{4}}, & u_4&=+\frac{\kappa_\infty}{\kappa_0}\frac{(t-\kappa_t\kappa_1)(t\kappa_t\kappa_1-1)}{T_{14}^\diamond T_{23}^\diamond}.
\end{align*}
A further local analysis shows that  $N_*$ is a double point on the curve and thus forms a singularity.
This finishes the proof of the proposition.
\end{proof}

In Remark \ref{rem:parametrisation}, we will obtain a rational parametrisation of the curve at infinity which further allows us to conclude that the curve at infinity is irreducible and that $N_*$ in Proposition \ref{prop:algebraic_characterisation} is the only singularity on it.



We now come to our main result, which shows that the Riemann-Hilbert correspondence becomes a rational mapping in $(f,g)\in\mathfrak{X}_t$ under the crystal limit.
\begin{theorem}\label{thm:main}
Upon fixing any $(f,g)\in \mathfrak{X}_t$, the Riemann-Hilbert correspondence defined in Definition \ref{def:rh}, evaluated at $(f,g)$, admits a power series expansion in $q$,
\begin{equation}\label{eq:RHexpansion}
     \mathrm{RH}_t(f,g)=\mathrm{RH}_t^\diamond(f,g)+\sum_{k=1}^\infty q^k R_k(f,g;t),
\end{equation}
which is absolutely convergent for small enough $q\in\mathbb{C}$, with  coefficients $R_k(f,g;t)$, $k\geq 1$, that are analytic in $(f,g)\in\mathfrak{X}_t$.\\
The leading order term,
\begin{equation}\label{eq:RH_leading_order}
    \mathrm{RH}_t^\diamond:\mathfrak{X}_t\rightarrow \mathcal{F}_t^\diamond, (f,g)\mapsto \eta^\diamond,
\end{equation}
is an isomorphism of algebraic varieties that admits the following explicit expression,
\begin{align*}
\eta_{12}^\diamond&=T_{12}^\diamond\,\frac{P_1(f,g)P_2(f,g)}{u\,f^2\, g}(f-\kappa_1)(f-\kappa_1^{-1}), & \eta_{34}^\diamond&=T_{34}^\diamond\,\frac{P_3(f,g)P_4(f,g)}{u\,f^2\, g},\\
\eta_{13}^\diamond&=T_{13}^\diamond\,\frac{P_1(f,g)P_3(f,g)}{u\,f^2\, g}(f-\kappa_1^{-1}),
 & \eta_{24}^\diamond&=T_{24}^\diamond\,\frac{P_2(f,g)P_4(f,g)}{u\,f^2\, g}(f-\kappa_1),\\
\eta_{14}^\diamond&=T_{14}^\diamond\,\frac{P_1(f,g)P_4(f,g)}{u\,f^2\, g}(f-\kappa_1), & 
\eta_{23}^\diamond&=T_{23}^\diamond\,\frac{P_2(f,g)P_3(f,g)}{u\,f^2\, g}(f-\kappa_1^{-1}),
\end{align*}
where the polynomials $P_k(f,g)$ are defined in equations \eqref{eq:Pdefi}, the coefficients $T_{ij}^\diamond$, $1\leq i< j\leq 4$, are given in equations \eqref{eq:Tcoefq0}, and the constant $u$ is given by
\begin{align*}
    u=&-t\,\kappa_0^{-2}\kappa_\infty^4(\kappa_0-1)^2(\kappa_t^2-1)(\kappa_1^2-1)(\kappa_\infty^2-1)^2\\
    &\cdot (\kappa_t\kappa_1t-1)
    (\kappa_t^{-1}\kappa_1t-1)
    (\kappa_t\kappa_1^{-1}t-1)
    (\kappa_t^{-1}\kappa_1^{-1}t-1).
\end{align*}
\end{theorem}

\begin{proof}
We start by fixing a real number $0<R<\min(U)$, where $U$ is the finite set
\begin{align*}
U=&\{|\kappa_0|^2,|\kappa_0|^{-2},|\kappa_t|^2,|\kappa_t|^{-2},|\kappa_1|^2,|\kappa_1|^{-2},|\kappa_\infty|^2,|\kappa_\infty|^{-2}\}\\
&\cup \left\{\left|t^{\epsilon}\kappa_t^{\epsilon_t}\,\kappa_1^{\epsilon_1}\right|:\epsilon,\epsilon_t,\epsilon_1\in\{\pm 1\}\right\}\\
&\cup \left\{\left|\kappa_0^{\epsilon_0}\,\kappa_t^{\epsilon_t}\,\kappa_1^{\epsilon_1}\,\kappa_\infty^{\epsilon_\infty}\right|:\epsilon_0,\epsilon_t,\epsilon_1,\epsilon_\infty\in\{\pm 1\}\right\}\\
&\cup \left\{\left|t^{\epsilon}\kappa_0^{\epsilon_0}\,\kappa_\infty^{\epsilon_\infty}\right|:\epsilon,\epsilon_0,\epsilon_\infty\in\{\pm 1\}\right\}.
\end{align*}
Due to assumptions \eqref{eq:non_rescrystal} and \eqref{eq:intro_irreducibleparametercrystal}, it follows that, for any $q\in\mathbb{C}$ with $0<|q|\leq R$, all the non-resonance conditions \eqref{eq:non_res} and non-splitting conditions \eqref{eq:intro_irreducibleparameter} are satisfied.  In particular, this means that, for any $(f,g)\in\mathfrak{X}_t$, $\mathrm{RH}_t(f,g)$ is well-defined for all $0<|q|\leq R$. By Proposition \ref{prop:rh}, we further know that $\mathrm{RH}_t(f,g)$ is analytic in $(f,g)\in\mathfrak{X}_t$ for fixed such $q$.

To study $\mathrm{RH}_t(f,g)$ in the limit $q\rightarrow 0$, we consider the different parts of its construction in the previous sections. Firstly, we apply Lemma \ref{lem:iso_init_to_linear}, and fix a representative coefficient matrix $A(z)\in\mathcal{A}_t$ of the image of $(f,g)$ under the mapping
  \eqref{eq:parametrisationmap}.

In Proposition \ref{prop:crystal_canonical_sol}, we derived analytic properties and asymptotic expansions as $q\rightarrow 0$ of the matrix functions $\Psi_0(z)$ and $\Psi_\infty(z)$, corresponding to the canonical solutions \eqref{eq:linear_sys_solutions} of
\begin{equation*}
    Y(qz)=A(z)Y(z).
\end{equation*}
For the corresponding connection matrix,
\begin{equation*}
    C(z)=\Psi_0(z)^{-1}\Psi_\infty(z),
\end{equation*}
these imply that $C(z)$ is an analytic function in $(z,q)\in\mathbb{C}^*\times D_R$, where $D_R:=\{q\in\mathbb{C}:|q|<R\}$. 

Recalling the notation in equation \eqref{eq:intro_xnotation}, this means that, for $1\leq k\leq 4$,
$C_{11}(x_k)$ and $C_{12}(x_k)$ are both analytic functions in $q$ on $D_R$. Thus their ratio
\begin{equation*}
\rho_k=\frac{C_{11}(x_k)}{C_{12}(x_k)},
\end{equation*}
is an analytic function 
\begin{equation*}
    \rho_k:D_R\rightarrow\mathbb{P}^1,
\end{equation*}
for $1\leq k\leq 4$.
Since we already know that the corresponding $\eta$-coordinates, $\eta=\mathrm{RH}_t(f,g)$, given in \eqref{eq:eta_defi}, are well-defined for all $0<|q|\leq R$, we therefore obtain that 
\begin{equation}\label{eq:mapF}
  D_R\setminus\{0\} \rightarrow \mathbb{C}^6,q\mapsto \eta,
\end{equation}
is an analytic map.

We now turn to the point $q=0$. The explicit values of the Tyurin parameters at $q=0$, are given in equation
\eqref{eq:rhocrystal}. The corresponding values of the $\eta$-coordinates at $q=0$ follow directly from these, and are given in the theorem. By direct inspection, one sees that these formulas are regular in $(f,g)$ on the whole of $\mathfrak{X}_t$. Thus $q=0$ is an apparent singularity of the mapping \eqref{eq:mapF}. That is, $\eta$ is analytic on the whole of $D_R$. 

This means that $\mathrm{RH}_t(f,g)=\eta$ is analytic in $q\in D_R$ for any fixed $(f,g)\in \mathfrak{X}_t$.
 On the other hand, for fixed $q\in D_R$, equal to zero or not, we also know that $\mathrm{RH}_t(f,g)$ is analytic in $(f,g)\in \mathfrak{X}_t$. By Hartogs' theorem, $\mathrm{RH}_t(f,g)$ is therefore analytic in $((f,g),q)$ on the whole of $\mathfrak{X}_t\times D_R$. In particular, it admits power series expansion \eqref{eq:RHexpansion} around $q=0$, that converges for $|q|<R$, for some coefficients $R_k(f,g;t)$, $k\geq 1$, that are analytic in $(f,g)\in\mathfrak{X}_t$.

What is left to prove, is that the leading order term $\mathrm{RH}_t^\diamond$ defines an isomorphism from $\mathfrak{X}_t$ to $\mathcal{F}_t^\diamond$. Firstly, since the coefficients in equations \eqref{eq:eta_equations}, which define $\mathcal{F}_t$, are analytic in $q$ at $q=0$, and we have shown that $\mathrm{RH}_t(f,g)$ is analytic in $q$ at $q=0$, it follows that $\mathrm{RH}_t^\diamond(f,g)\in \mathcal{F}_t^\diamond$, for all $(f,g)\in\mathfrak{X}_t$. Thus, $\mathrm{RH}_t^\diamond$ indeed defines a mapping from $\mathfrak{X}_t$ to $\mathcal{F}_t^\diamond$. We note that one can also check this directly using the explicit formulas for $\mathrm{RH}_t^\diamond$ in the theorem.

Furthermore, since $\mathrm{RH}_t(f,g)$ is analytic in $((f,g),q)$ on the whole of $\mathfrak{X}_t\times D_R$, it follows in particular that $\mathrm{RH}_t^\diamond$ is a regular map from $\mathfrak{X}_t$ to $\mathcal{F}_t^\diamond$. Again, one can also check this directly using the explicit formulas for $\mathrm{RH}_t^\diamond$ in the theorem.

Next, we show that $\mathrm{RH}_t^\diamond$ is a bijective mapping. To this end, note that $\mathrm{RH}_t^\diamond$ is given by the composition of the following four mappings,
\begin{align}
  &&   &\mathfrak{X}_t    \,\rightarrow\, \mathcal{A}_t/\hspace{-1mm}\sim, &  &(f,g) \mapsto [A(z)], &&\label{eq:map1}\\
&&     &\mathcal{A}_t/\hspace{-1mm}\sim \,\rightarrow\, \mathcal{M}_t^\diamond, & &[A(z)] \mapsto [C(z)], &&\label{eq:map2}\\
&&    &\mathcal{M}_t^\diamond \,\rightarrow\, \mathcal{S}_t^\diamond, & &[C(z)] \mapsto [\rho], &&\label{eq:map3}\\
&&    &\mathcal{S}_t^\diamond \,\rightarrow\, \mathcal{F}_t^\diamond, & &[\rho] \mapsto \eta, && \label{eq:map4}
\end{align}
where $\mathcal{S}_t^\diamond\subseteq (\mathbb{P}^1)^4/\mathbb{C}^*$ denotes the space of all $\rho\in(\mathbb{P}^1)^4$, modulo overall scalar multiplication, that satisfy \eqref{eq:tyurineq} and \eqref{eq:tyurinnonzero} with $q=0$.

It follows from Lemma \ref{lem:iso_init_to_linear} that the first mapping, \eqref{eq:map1}, is bijective.

Regarding the second mapping, \eqref{eq:map2}, we recall that any representative $A(z)$ of an equivalence class $[A(z)]\in\mathcal{A}_t/\hspace{-1mm}\sim$, is unique up to conjugation $A(z)\mapsto D_1 A(z) D_1^{-1}$ by an invertible diagonal matrix $D_1$. Such a conjugation affects the diagonalising matrix $H$ in \eqref{eq:diagonal} by left-multiplication $H\mapsto D_1 H$. Furthermore, $H$ is itself only defined up to arbitrary right-multiplication $H\mapsto H D_2$ by (invertible) diagonal matrices.
Thus the corresponding connection matrix $C(z)$ is uniquely defined up to
\begin{equation*}
C(z)=t^{-1}\kappa_0^{-\sigma_3}H^{-1}A(z)\mapsto D_2^{-1} C(z) D_1^{-1},
\end{equation*}
for arbitrary invertible diagonal matrices $D_1,D_2$. This fits the definition of $\mathcal{M}_t^\diamond$ perfectly and, in particular, \eqref{eq:map2} is a well-defined mapping. It is now elementary to check that it is furthermore a bijection.

Proving that the third mapping, \eqref{eq:map3}, is bijective is done analogously to the proof of \cite[Theorem 2.15]{jr_qp6}. Similarly, proving that the fourth mapping, \eqref{eq:map4}, is bijective is done as in the proof of \cite[Theorem 2.20]{jr_qp6}.

We conclude that $\mathrm{RH}_t^\diamond$ is a regular, bijective rational mapping from $\mathfrak{X}_t$ to $\mathcal{F}_t^\diamond$. As $\mathcal{F}_t^\diamond$ is smooth and hence normal, it follows from the ``original form" of Zariski's main theorem \cite[III,\S 9] {mumfordred} that, to establish that $\mathrm{RH}_t^\diamond$ is an isomorphism, it is enough to show that it has a rational inverse.

Thus, what remains to be done, is the rational reconstruction of a coefficient matrix $A(z)$ from coordinate-values $\eta\in \mathcal{F}_t^\diamond$. We will do this using the technology of Mano decompositions, first developed in \cite{ohyamaramissualoy}, following \cite[\S 4.3]{rofqpviasymptotics}, in the degenerate case when $q=0$.

We consider a connection matrix given as a product
\begin{equation}\label{eq:connect}
    C(z)=C^i(z) C^e(z),
\end{equation}
where the individual factors $C^i(z)$ and $C^e(z)$ are degree 1 matrix polynomials
\begin{equation*}
    C^i(z)=C_0^i+z\,C_1^i,\quad  C^e(z)=C_0^e+z\,C_1^e,
\end{equation*}
which, inspired by Definition \ref{def:monodromy_manifold_crystal}, satisfy
\begin{equation*}
    C_0^i=-t\,\kappa_0^{\sigma_3}C_1^i \lambda^{\sigma_3},\quad
    C_0^e=-\lambda^{-\sigma_3}C_1^e \kappa_\infty^{-\sigma_3},
\end{equation*}
and
\begin{equation*}
    |C^i(z)|=c_i (z-\kappa_t t)(z-\kappa_t^{-1}t),\quad 
       |C^e(z)|=c_e (z-\kappa_1)(z-\kappa_1^{-1}),
\end{equation*}
for some immaterial constants $c_i,c_e\in\mathbb{C}^*$ and a yet to be determined scalar $\lambda\in\mathbb{C}^*$.

We set
\begin{equation*}
    C^e(z)=\begin{bmatrix}
    (1-\kappa_1\,\lambda^{+1}\kappa_\infty^{-1})(1-z\,\lambda^{+1}\kappa_\infty^{+1}) &  (1-\kappa_1\,\lambda^{+1}\kappa_\infty^{+1})(1-z\,\lambda^{+1}\kappa_\infty^{-1})\\
     (1-\kappa_1\,\lambda^{-1}\kappa_\infty^{-1})(1-z\,\lambda^{-1}\kappa_\infty^{+1}) &  (1-\kappa_1\,\lambda^{-1}\kappa_\infty^{+1})(1-z\,\lambda^{-1}\kappa_\infty^{-1})
    \end{bmatrix},
\end{equation*}
so that all the above properties for $C^e(z)$ are satisfied.

We now consider the ratio of the Tyurin parameters $\rho_3$ and $\rho_4$ of $C(z)$,
\begin{align*}
\frac{\rho_3}{\rho_4}&=\frac{\pi[C(\kappa_1^{+1})]}{\pi[C(\kappa_1^{-1})]}
=\frac{\pi[C^i(\kappa_1^{+1})C^e(\kappa_1^{+1})]}{\pi[C^i(\kappa_1^{-1})C^e(\kappa_1^{-1})]}\\
&=\frac{\pi[C^e(\kappa_1^{+1})]}{\pi[C^e(\kappa_1^{-1})]}=\frac{(\lambda-\kappa_1^{+1}\kappa_\infty^{+1})(\lambda-\kappa_1^{-1}\kappa_\infty^{-1})}{(\lambda-\kappa_1^{-1}\kappa_\infty^{+1})(\lambda-\kappa_1^{+1}\kappa_\infty^{-1})},
\end{align*}
where in the third equality we used that $\pi[\cdot]$ is invariant under left-multiplication by invertible matrices.

We fix an $\eta^*\in \mathcal{F}_t^\diamond$ in sufficiently generic position and our aim is to choose $\lambda$ and $C^i(z)$ such that the $\eta$-coordinates of $C(z)$, given in \eqref{eq:connect}, equal $\eta^*$.
In particular, we require 
\begin{equation}\label{eq:p34}
    \frac{(\lambda-\kappa_1^{+1}\kappa_\infty^{+1})(\lambda-\kappa_1^{-1}\kappa_\infty^{-1})}{(\lambda-\kappa_1^{-1}\kappa_\infty^{+1})(\lambda-\kappa_1^{+1}\kappa_\infty^{-1})}=\frac{\rho_3}{\rho_4}=\frac{T_{14}^\diamond\,\eta_{13}^*}{T_{13}^\diamond\,\eta_{14}^*}.
\end{equation}
This equation has two solutions $\lambda\in\mathbb{C}^*$, related by the involution $\lambda\mapsto \lambda^{-1}$, and we fix $\lambda$ to be one of the two.

We now set
\begin{equation*}
    C^i(z)=\begin{bmatrix}
    r\,(1-\kappa_t\,\lambda^{+1}\kappa_0^{-1})(1-z\,\lambda^{-1}\kappa_0^{-1}t^{-1}) &   (1-\kappa_t\,\lambda^{-1}\kappa_0^{-1})(1-z\,\lambda^{+1}\kappa_0^{-1}t^{-1})\\
     r\,(1-\kappa_t\,\lambda^{+1}\kappa_0^{+1})(1-z\,\lambda^{-1}\kappa_0^{+1}t^{-1}) &   (1-\kappa_t\,\lambda^{-1}\kappa_0^{+1})(1-z\,\lambda^{+1}\kappa_0^{+1}t^{-1})
    \end{bmatrix},
\end{equation*}
and choose $r\in\mathbb{C}^*$ exactly such that
\begin{equation}\label{eq:p24}
 \frac{\pi[C(\kappa_t^{-1}t)]}{\pi[C(\kappa_1^{-1})]}=\frac{\rho_2}{\rho_4}=\frac{T_{14}^\diamond\,\eta_{12}^*}{T_{12}^\diamond\,\eta_{14}^*}.   
\end{equation}
A direct calculation yields the following explicit formula for $r$,
\begin{align*}
    r=-&\frac{1}{t\,\kappa_t\,\kappa_1\,\lambda^2}\cdot \frac{(\lambda-\kappa_0\,\kappa_t)(\lambda-\kappa_0^{-1}\kappa_t)(\lambda-\kappa_1\,\kappa_\infty)(\lambda-\kappa_1\,\kappa_\infty^{-1})}{(\lambda-\kappa_0\,\kappa_t^{-1})(\lambda-\kappa_0^{-1}\kappa_t^{-1})(\lambda-\kappa_1^{-1}\kappa_\infty)(\lambda-\kappa_1^{-1}\kappa_\infty^{-1})}\\
    &\cdot 
    \frac{(\lambda-t\,\kappa_t^{-1}\kappa_\infty^{-1})(\lambda-\kappa_1^{-1}\kappa_\infty)\rho_{24}-(\lambda-t\,\kappa_t^{-1}\kappa_\infty)(\lambda-\kappa_1^{-1}\kappa_\infty^{-1})}
    {(\lambda-t^{-1}\kappa_t\,\kappa_\infty)(\lambda-\kappa_1\,\kappa_\infty^{-1})\rho_{24}-(\lambda-t^{-1}\kappa_t\,\kappa_\infty^{-1})(\lambda-\kappa_1\,\kappa_\infty)},
\end{align*}
where we used short-hand notation
\begin{equation*}
    \rho_{24}:=\frac{T_{14}^\diamond\,\eta_{12}^*}{T_{12}^\diamond\,\eta_{14}^*}.
\end{equation*}

We have now constructed a connection matrix $C(z)$, given by the product \eqref{eq:connect},
that satisfies the desired properties (1)',(2)' and (3)' in Definition \ref{def:monodromy_manifold_crystal}, and whose $\eta$-coordinates satisfy, see equations \eqref{eq:p34} and \eqref{eq:p24},
\begin{equation*}
    \frac{T_{14}^\diamond\,\eta_{13}}{T_{13}^\diamond\,\eta_{14}}=\frac{\rho_3}{\rho_4}=\frac{T_{14}^\diamond\,\eta_{13}^*}{T_{13}^\diamond\,\eta_{14}^*},\qquad
  \frac{T_{14}^\diamond\,\eta_{12}}{T_{12}^\diamond\,\eta_{14}}=  \frac{\rho_2}{\rho_4}=\frac{T_{14}^\diamond\,\eta_{12}^*}{T_{12}^\diamond\,\eta_{14}^*}.
\end{equation*}
Since $\eta^*$ was chosen in generic position, the intersection of the hyperplanes $\eta_{13} \eta_{14}^*-\eta_{14} \eta_{13}^*=0$, $\eta_{12} \eta_{14}^*-\eta_{14} \eta_{12}^*=0$ and $\mathcal{F}_t^\diamond$ consists of only one point, $\eta^*$, so 
that the $\eta$-coordinates of the connection matrix must equal $\eta^*$.

The connection matrix $C(z)$, however, depends on the choice of solution $\lambda$ of the degree two algebraic equation \eqref{eq:p34}. Switching to the other solution, $\lambda\mapsto \lambda^{-1}$, transforms $C^e(z)$ and $C^i(z)$ as follows,
\begin{equation*}
    C^e(z)\mapsto \sigma_1\, C^e(z),\quad
    C^i(z)\mapsto  r^{-1}C^i(z)\,\sigma_1,\quad \sigma_1:=\begin{bmatrix}
        0 & 1\\
        1 & 0
    \end{bmatrix},
\end{equation*}
and thus transforms $C(z)$ as
\begin{equation*}
    C(z)\mapsto r^{-1}C(z).
\end{equation*}
So the entries of the coefficients of $C(z)$ are algebraic in $\eta^*$.

Next, we define a corresponding coefficient matrix $A(z)$. To this end,
write
\begin{equation*}
    C(z)=C_0+z\,C_1+z^2C_2,
\end{equation*}
and define a matrix $H$ by
\begin{equation*}
    H=t^{-1}\,\kappa_\infty^{\sigma_3} \,C_2^{-1}\kappa_0^{-\sigma_3}.
\end{equation*}
We then define a corresponding coefficient matrix by
\begin{equation*}
    A(z):=H\,t\,\kappa_0^{\sigma_3}C(z)\in \mathcal{A}_t,
\end{equation*}
where we note that $H$ was chosen such that the coefficient of $z^2$ is correctly normalised to be $\kappa_\infty^{\sigma_3}$.

Crucially, under the involution $\lambda\mapsto \lambda^{-1}$, $H$ transforms as $H\mapsto r\,H$, and thus $A(z)$ is completely invariant. That is, the coefficients of $A(z)$ are matrices of rational functions in $\eta^*\in \mathcal{F}_t^\diamond$. Through the isomorphism \eqref{eq:parametrisationmap}, see Lemma \ref{lem:iso_init_to_linear}, we thus obtain rational formulas $f=f(\eta^*)$ and $g=g(\eta^*)$, which, by construction, form a rational inverse of $\mathrm{RH}_t^\diamond$.
This shows that $\mathrm{RH}_t^\diamond$ is an isomorphism and completes the proof of the theorem.
\end{proof}

\begin{remark}\label{rem:parametrisation}
Recall that the open surface $\mathfrak{X}_t$ is constructed by removing the strict transform of the union of curves $f=0$, $f=\infty$, $g=0$, $g=\infty$ from the compact rational surface $\overline{\mathfrak{X}}_t$. The mapping $\mathrm{RH}_t^\diamond$ extends uniquely to a regular rational mapping from $\overline{\mathfrak{X}}_t$ into the smooth (projective) Segre surface $\overline{\mathcal{F}}_t^\diamond$. This extension maps the strict transforms of the curves $f=0$, $f=\infty$ and $g=\infty$ onto the singularity $N_*$ in the curve at infinity, see Proposition \ref{prop:algebraic_characterisation}. On the other hand, when restricted to the strict transform of $g=0$, $\mathrm{RH}_t^\diamond$ is given by
\begin{equation}\label{eq:param}
\begin{aligned}
    N_\infty&=0,\\
    N_{12}&=T_{12}^\diamond(f-\kappa_1)(f-\kappa_1^{-1})(f-t\,\kappa_t\,\kappa_\infty^2)(f-t\,\kappa_t^{-1}\kappa_\infty^2),\\
    N_{13}&=T_{13}^\diamond(f-t\,\kappa_t^{-1})(f-\kappa_1^{-1})(f-t\,\kappa_t\,\kappa_\infty^2)(f-\kappa_1\,\kappa_\infty^2),\\
    N_{14}&=T_{14}^\diamond(f-\kappa_1)(f-t\,\kappa_t^{-1})(f-t\,\kappa_t\,\kappa_\infty^2)(f-t\,\kappa_1^{-1}\kappa_\infty^2),\\
    N_{23}&=T_{23}^\diamond(f-t\,\kappa_t)(f-\kappa_1^{-1})(f-t\,\kappa_t^{-1}\kappa_\infty^2)(f-\kappa_1\,\kappa_\infty^2),\\
    N_{24}&=T_{24}^\diamond(f-\kappa_1)(f-t\,\kappa_t)(f-t\,\kappa_t^{-1}\kappa_\infty^2)(f-\kappa_1^{-1}\kappa_\infty^2),\\
    N_{34}&=T_{34}^\diamond(f-t\,\kappa_t)(f-t\,\kappa_t^{-1})(f-\kappa_1\kappa_\infty^2)(f-\kappa_1^{-1}\kappa_\infty^2).
\end{aligned}
\end{equation}
This defines a regular map
\begin{equation}\label{eq:rational_param}
    \{f\in\mathbb{CP}^1\}\rightarrow \overline{\mathcal{F}}_t^\diamond\setminus \mathcal{F}_t^\diamond,
\end{equation}
which is easily seen to be injective when restricted to $\mathbb{C}^*$; only the pair of points $f=0$ and $f=\infty$ have the same image, namely the singularity $N_*$. These are also the only two points mapped to $N_*$. 
The image of the mapping \eqref{eq:rational_param} is an irreducible subset of the curve at infinity. If the curve at infinity were reducible, then any of its irreducible components would have degree less than four, which contradicts the fact that one of them necessarily contains the image of the mapping \eqref{eq:rational_param}. It follows that the curve at infinity is irreducible and that the mapping \eqref{eq:rational_param} is surjective. In particular, \eqref{eq:rational_param} defines a rational parametrisation of the curve at infinity of the Segre surface. It further follows that $N_*$ is the only singularity on the curve at infinity. So, the curve at infinity is a singular, rational, irreducible, quartic curve with a unique singularity.
\end{remark}

%% file: conclusion.tex
\section{Conclusion}\label{s:conc}
In this paper, we studied the Riemann-Hilbert correspondence for $q\Psix$ in the limit $q\rightarrow 0$ and obtained three results: first, that the correspondence becomes a bi-rational mapping, second, that there is an explicit formula for the limiting mapping, and, third, that it is an isomorphism.

Although we focused on a particular $q$-difference Painlev\'e equation, we expect that similar results arise for a wide class of $q$-difference equations, including all Fuchsian $q$-difference systems such as $q$-Garnier systems \cite{sakai_garnier}. We expect that the results may also be true for irregular linear problems, such as those found in \cite{murata} for other $q$-Painlev\'e equations and their higher-order analogues. 

Whilst we only studied the leading-order term in the Riemann-Hilbert correspondence in the crystal limit, the question of the explicit determination
of later terms in the asymptotic expansion (cf equation \eqref{eq:RHexpansion}) is an interesting open question for future research. The CFT approach \cite{jimbonagoyasakai} to $q\Psix$ may be particularly relevant in this regard.

Another open question is whether there exist other scalings of $t$ and $q$ that give rise to different limits of Riemann-Hilbert problems as $q\rightarrow 0$. In particular, we note that ultra-discrete Painlev\'e equations arise in limits such as $t=e^{-T/\epsilon}$, $q=e^{-Q/\epsilon}$, $\epsilon\rightarrow 1$, with $T$ and $Q$ fixed. Investigations of the Riemann-Hilbert correspondence in such limits could also form interesting future research directions.

%% file: main.bbl
\begin{bibdiv}
 \begin{biblist}
 








\bib{lecaine}{article}{
   author={Le Caine, J.},
   title={{The linear $q$-difference equation of the second order}},
   journal={Amer. J. Math.},
   volume={65},
   date={1943},
   pages={585--600}
}


\bib{carmichael1912}{article}{
  author= {Carmichael,R.D.},
  title={The general theory of linear $q$-difference equations},
  journal={American Journal of Mathematics},
  volume={34},
  pages={147--168},
  year={1912}
  }





\bib{fokasitsbook}{book}{
   author={Fokas, A.S.},
   author={Its, A.R.},
   author={Kapaev, A.A.},
   author={Novokshenov, V.Y.},
   title={Painlev\'e{} transcendents},
   series={Mathematical Surveys and Monographs},
   volume={128},
   note={The Riemann-Hilbert approach},
   publisher={American Mathematical Society, Providence, RI},
   date={2006},
   pages={xii+553}
}


\bib{inaba2006}{article}{
   author={Inaba, M.A.},
   author={Iwasaki, K.},
   author={Saito, M.H.},
   title={Dynamics of the sixth Painlev\'e{} equation},
   conference={
      title={Th\'eories asymptotiques et \'equations de Painlev\'e},
   },
   book={
      series={S\'emin. Congr.},
      volume={14},
      publisher={Soc. Math. France, Paris},
   },
   isbn={978-2-85629-229-7},
   date={2006},
   pages={103--167}
}



\bib{itsnovokbook}{book}{
   author={Its, A.R.},
   author={Novokshenov, V.Y.},
   title={The isomonodromic deformation method in the theory of Painlev\'e{}
   equations},
   series={Lecture Notes in Mathematics},
   volume={1191},
   publisher={Springer-Verlag, Berlin},
   date={1986},
   pages={iv+313}
}

\bib{jimbo82}{article}{
   author={Jimbo, M.},
   title={Monodromy problem and the boundary condition for some Painlev\'e{}
   equations},
   journal={Publ. Res. Inst. Math. Sci.},
   volume={18},
   date={1982},
   number={3},
   pages={1137--1161}
}




  
\bib{jimbonagoyasakai}{article}{
    AUTHOR = {Jimbo, M.},
    AUTHOR = {Nagoya, H.},
    AUTHOR = {Sakai, H.},
     TITLE = {C{FT} approach to the {$q$}-{P}ainlev\'{e} {VI} equation},
   JOURNAL = {J. Integrable Syst.},
    VOLUME = {2},
      YEAR = {2017},
    NUMBER = {1},
     PAGES = {27}
} 

\bib{jimbosakai}{article}{
  author={Jimbo, M.},
  author={Sakai, H.},
  title={A $q$-analogue of the sixth {P}ainlev{\'e} equation},
  journal={ Lett. Math. Phys.},
  volume={38},
  pages={145--154},
  year={1996}
  }

\bib{jmr_segre}{article}{
  author={Joshi, N.},
author={Mazzocco, M.},
  author={Roffelsen, P.},
  title={Segre surfaces and geometry of the Painlev\'e equations},
journal={arXiv:2405.10541 [math-ph]},
    year={2024}
    }




\bib{jr_qp6}{article}{
  author={Joshi, N.},
  author={Roffelsen, P.},
  title={On the monodromy manifold of q-Painlev\'e VI and its Riemann-Hilbert problem},
  journal={Commun. Math. Phys.},
  volume={404},
  pages={97–-149},
  year={2023}
  }



\bib{kashiwara}{article}{
   author={Kashiwara, M.},
   title={Crystalizing the $q$-analogue of universal enveloping algebras},
   journal={Comm. Math. Phys.},
   volume={133},
   date={1990},
   number={2},
   pages={249--260}
}


\bib{mumfordred}{book}{
   author={Mumford, D.},
   title={The red book of varieties and schemes},
   series={Lecture Notes in Mathematics},
   volume={1358},
   publisher={Springer-Verlag, Berlin},
   date={1999}
}




\bib{murata}{article}{
   author={Murata, M.},
   title={Lax forms of the $q$-Painlev\'{e} equations},
   journal={J. Phys. A},
   volume={42},
   date={2009},
   number={11},
   pages={115201, 17}
}







\bib{ohyamaramissualoy}{article}{
    AUTHOR = {Ohyama, Y.},
    AUTHOR = {Ramis, J.P.},
    AUTHOR = {Sauloy, J.},
     TITLE = {The space of monodromy data for the {J}imbo-{S}akai family of
              {$q$}-difference equations},
   JOURNAL = {Ann. Fac. Sci. Toulouse Math. (6)},
    VOLUME = {29},
      YEAR = {2020},
    NUMBER = {5},
     PAGES = {1119--1250},
}

\bib{ormerod}{article}{
   author={Ormerod, C.M.},
   title={The lattice structure of connection preserving deformations for
   $q$-Painlev\'{e} equations I},
   journal={SIGMA Symmetry Integrability Geom. Methods Appl.},
   volume={7},
   date={2011}}



\bib{putsaito}{article}{
   author={van der Put, M.},
   author={Saito, M.H.},
   title={Moduli spaces for linear differential equations and the
   Painlev\'e{} equations},
   journal={Ann. Inst. Fourier (Grenoble)},
   volume={59},
   date={2009},
   number={7},
   pages={2611--2667}
}


 \bib{rofqpviasymptotics}{article}{
  author={Roffelsen, P.},
  title={On $q$-Painlev\'e VI and the geometry of Segre surfaces},
  journal={To appear in Nonlinearity},
  date={2024}
}   




\bib{s:01}{article}{
  author={Sakai, H.},
  title={Rational surfaces associated with affine root systems
      and geometry of the {P}ainlev\'e equations},
 journal={Commun. Math. Phys.},
  volume={220},
  pages={165--229},
  date={2001}
}

\bib{sakai_garnier}{article}{
   author={Sakai, H.},
   title={A $q$-analog of the Garnier system},
   journal={Funkcial. Ekvac.},
   volume={48},
   date={2005},
   number={2},
   pages={273--297}
}


\end{biblist}
\end{bibdiv}
